\documentclass[hidelinks, 11 pt]{article}
\usepackage{amsfonts}
\usepackage{amsmath}
\usepackage{amssymb}
\usepackage{amsthm}
\usepackage{subfigure}
\usepackage[pdftex]{graphicx}
\usepackage{natbib}
\usepackage[hypertexnames=false]{hyperref}
\usepackage{setspace}
\usepackage{pdfsync}
\usepackage{lscape}
\usepackage[hmargin=1.25in,vmargin=1.25in]{geometry}

\hypersetup{colorlinks, linkcolor = {teal}, citecolor = {blue}, urlcolor ={blue!80!black}}

\makeatletter
\renewcommand\paragraph{\@startsection{paragraph}{4}{\z@}%
            {-2.5ex\@plus -1ex \@minus -.25ex}%
            {1.25ex \@plus .25ex}%
            {\normalfont\normalsize\bfseries}}
\makeatother
\usepackage{mathrsfs,dsfont}
\usepackage[dvipsnames,svgnames, x11names]{xcolor}
\usepackage{pgfplots}
\pgfplotsset{compat=1.17} 
\usetikzlibrary{patterns}
\usepackage{caption}

\DeclareGraphicsExtensions{.jpg}

\def\qed{\rule{2mm}{2mm}}

\newtheorem{condition}{Condition}

\newtheorem{theorem}{Theorem}[section]
\newtheorem{lemma}{Lemma}[section]
\newtheorem{definition}{Definition}[section]

\newtheorem{corollary}{Corollary}[section]
\newtheorem{remark}{Remark}[section]

\providecommand{\keywords}[1]{   
  \textbf{Keywords:} #1
}

\providecommand{\jel}[1]{  
  \textbf{JEL classification codes:} #1
}

\title{Approximate Sparsity Class and Minimax Estimation}
\author{Lucas Z. Zhang\\ Department of Economics\\University of California, Los Angeles\\\href{mailto:lucaszz@g.ucla.edu}{lucaszz@g.ucla.edu}}
\date{\today}

\begin{document}

\maketitle
\allowdisplaybreaks 

\begin{abstract}
Motivated by the orthogonal series density estimation in $L^2([0,1],\mu)$, in this project we consider a new class of functions that we call the approximate sparsity class. This new class is characterized by the rate of decay of the individual Fourier coefficients for a given orthonormal basis. We establish the $L^2([0,1],\mu)$ metric entropy of such class, with which we show the minimax rate of convergence. For the density subset in this class, we propose an adaptive density estimator based on a hard-thresholding procedure that achieves this minimax rate up to a $\log$ term.
\end{abstract}


\jel{C1; C13; C14}

\section{Introduction}

First introduced by \cite{cencov62,kromaltarter68,schwartz67,watson69}, density estimation using orthogonal series has since been extensively studied. Suppose we observe an i.i.d sample $\{X_i\}_{i=1}^n$ from the distribution of a random variable $X$ on $[0,1]$ with probability density $f_X$. Let $\{\phi_j\}_{j=1}^\infty$ be an orthonormal basis of $L^2([0,1],\mu)$. If $f_X\in L^2([0,1],\mu)$, then it enjoys an expansion $f_X(\cdot) = \sum_{j=1}^\infty\theta_j\phi_j(\cdot)$, where $\theta_j = E\left[\phi_j(X)\right]$ is the j-th Fourier coefficients for all $j\geq 1$. A natural estimator then takes the form
\begin{equation}
\hat{f}_J(\cdot) := \sum_{j=1}^J\hat\theta_j\phi_j(\cdot),\quad\hat\theta_j = \frac{1}{n}\sum_{i=1}^n\phi_j(X_i).
\end{equation}
When evaluating the performance of such estimator using MISE criteria, the cutoff value $J$ plays the role of a tuning parameter, which, when chosen properly, balances the variance and bias and hence minimizes MISE. The optimal choices of the cutoff $J$ have been discussed extensively in the literature, see, for example, \cite{hall87,hart85,kromaltarter76,watson69}. A generalized version of such estimator is based on thresholding
\begin{equation}\label{threshold}
\hat{f}(\cdot) := \sum_{j=1}^\infty \omega_j\hat\theta_j\phi_j(\cdot),
\end{equation}
where $\omega_j$'s are the so-called thresholding parameters and typically $\omega_j \in [0,1]$ plays a role of shrinking the estimated coefficients $\hat\theta_j$'s. Note that the previous estimator $\hat{f}_J(x)$ is a special case of the thresholding estimator (\ref{threshold}) when we use $\omega_j = \mathbf{1}\{j\leq J\}$. The thresholding estimator of the form (\ref{threshold}) was first considered by \cite{kromaltarter68}, and many thresholding procedures have since been extensively studied, see for example, \cite{bunea2010spades,chickencai05,digglehall86,donoho96,efro86,Wahba81} and the references therein.

Various previous works have shown that, when the thresholding parameters $\omega_j$'s are chosen properly, the orthogonal series estimators of the form (\ref{threshold}) can achieve minimax rates of convergence over familiar function classes such as Sobolev ellipsoids and Besov spaces (commonly discussed in the context of wavelets thresholding), see for example, \cite{bunea2010spades,chickencai05,donoho96,donoho98,efro86,hall86,hardle2012wavelets} and the references therein. We note that many function classes considered, such as the Sobolev ellipsoids, are characterized by the restrictions on the Fourier coefficients of the functions in those classes. Those restrictions on the Fourier coefficients in turn help researchers establish the statistical properties of estimators of the form (\ref{threshold}).  

The type of restriction of Fourier coefficients that we are interested in concerns how fast Fourier coefficients decay. It is motivated, for instance, by the discussion in \cite{efro2008,hall86} that for a twice differentiable density $f$ on $[0,1]$ with bounded second derivative, its Fourier expansion $f(\cdot) = \sum_{j=1}^\infty \theta_j\phi_j(\cdot)$ given the cosine orthonormal basis $\{\phi_j\}_{j=1}^\infty$ has the property that the Fourier coefficients $\{\theta_j\}$ decay at rate $j^{-2}$. Similar results concerning Hermite orthonormal basis were discussed in \cite{schwartz67}. In this paper, we generalize such restrictions. In particular, we allow for that: (i) the non-increasing reordering (in absolute value) of the Fourier coefficients satisfies that the j-th largest Fourier coefficient decays at rate $|\theta_{(j)}|\leq Aj^{-k}$ for some constants $A,k$; (ii) the tail sum of the Fourier coefficients satisfies $\sum_{j=J+1}^\infty \theta_j^2\leq CJ^{-2k+1}$ for some constant $C$. We call such class the ``approximate sparsity class'', motivated by the ``approximate sparsity condition'' from \cite{BCCHKN}. This class is interesting for the following reasons. First, it generalizes various previously considered classes of functions that are characterized by their Fourier coefficients, including but not limited to the Sobolev ellipsoids and H\"{o}lder classes (e.g. \cite{Katznelson04}). Second, in practice, researchers may be uncertain about the order (in terms of absolute magnitude) of the true coefficients, and our approximate sparsity class reflects such uncertainty while still maintaining the decaying properties of the re-ordered coefficients that are important for estimation purposes. 

While the approximate sparsity class may seem complex, we will show that such class is sandwiched in between two classes with simpler structures, and we will use sandwich arguments to formally establish the $L^2([0,1],\mu)$ $\epsilon$-metric entropy of the approximate sparsity class and of its density subclass. To establish the upper bound on the entropy, we use an existing result from the full approximation set (\cite{lorentz66}). On the other hand, we prove a lower bound using a volume-type argument inspired by \cite{Smolyak60}. With these entropy bounds, we can apply the results from \cite{yangbarron99} to establish the minimax rate of convergence (in terms of MISE) for both density estimation and nonparametric regression with Gaussian noise. We obtain a minimax rate of order $n^{-(2k-1)/(2k)}$. As mentioned before, one can verify that the Sobolev ellipsoids are subsets of our approximate sparsity class, and as expected, this rate on approximate sparsity class is slower than the well-known minimax rate $n^{-2k/(2k+1)}$ for Sobolev ellipsoids. For a comprehensive review on the minimax estimation and the connections between metric entropy and minimax rates, see for example, \cite{tsybakov2008intro,yangbarron97,yangbarron99} and the references therein. 

With the obtained minimax rate in mind, we propose an adaptive density estimator based on a data-driven hard thresholding procedure. The main idea is as follows. We first pick a large cutoff $J$, potentially much larger than sample size $n$, and estimate the first $J$ Fourier coefficients by the sample mean $\hat\theta_j = n^{-1}\sum_{i=1}^n\phi_j(X_i)$. However, including all the $J$ terms in the series will inevitably lead to large variance in estimation and hence suboptimal rate. To overcome this issue, in the second step we use a hard-thresholding procedure to select all the $\hat\theta_j$ above a certain data-driven threshold $\lambda$ and penalize the rest to zero. Then we form our estimator 
\begin{equation}
\tilde{f}_J(\cdot) = \sum_{j\leq J, |\hat\theta_j|\geq \lambda}\hat\theta_j\phi_j(\cdot).
\end{equation}
In the final step, we project $\tilde{f}_J$ onto the space of densities to obtain a bona-fide density. 

Note that since $\{\phi_j\}_{j=1}^\infty$ is an orthonormal basis, the estimation errors will be characterized entirely by Fourier coefficients $\theta_j =E\left[\phi_j(X_i)\right]$ and estimated $\hat\theta_j = n^{-1}\sum_{i=1}^n \phi_j(X_i)$. This suggests that essentially we are dealing with the problem of estimating and selecting many approximate means. Following the ideas from \cite{BCCHKN}, we pick the threshold $\lambda$ to be greater than the $(1-\alpha)$-quantile of $\max_{1\leq j\leq J}|\hat\theta_j-\theta_j|$, with $\alpha\downarrow 0$ as $n\to\infty$. Previous literature suggests that such $\lambda$ can be approximated in a data-driven way using results from self-normalized moderate deviation theories or high-dimensional bootstrap, see, for example, \cite{BCCH,BCCHKN} and the references therein. Although these constructions of $\lambda$ can be used to show the non-asymptotic rate, unfortunately, they are not sufficient for establishing the rate in terms of MISE. Instead, we modify their constructions and propose an alternative data-driven $\lambda$, and we show that such $\lambda$ has the desired property with the help of the Talagrand's inequality (see e.g. \cite{Bousquet2003Talagrand}). We then show that if the true density $f$ belongs to an approximate sparsity class, our proposed hard-thresholding estimator is nearly minimax up to a log factor. Moreover, the estimator itself does not depend on the assumptions on the parameters of the sparsity class and is therefore adaptive.

The remaining of the paper is organized as follows. In the next section, we introduce notations and formally define the approximate sparsity class. In Section 3 we establish the metric entropy and minimax rates for density estimation and nonparametric regression with Gaussian noise in such classes. In Section 4 we elaborate on the aforementioned adaptive density estimator and we derive its rate of convergence. We then provide a specific example using the cosine basis for twice differentiable densities, in which we verify the assumptions and establish the rate of convergence by applying the results from the main theorem. We conduct simulation studies in Section 5 to illustrate the performance of our estimator, and we conclude in Section 6. The proofs are deferred to the appendix.

\section{Approximate Sparsity Class}
Suppose $\Phi:= \{\phi_j\}_{j=1}^\infty$ is an orthonormal basis of $L^2(\mathcal{X},\mu)$ for $\mathcal{X}\subset \mathbf{R}$, and without loss of generality, let $\mathcal{X} = [0,1]$. Then, for any $i\neq j$,
\begin{equation}
\int_0^1 \phi_j^2(x) d\mu(x) = 1,\quad\int_0^1 \phi_i(x)\phi_j(x) d\mu(x)  = 0
\end{equation}
Moreover, for any $f\in L^2([0,1],\mu)$, there is a representation
\begin{equation}
f(\cdot) = \sum_{j=1}^\infty\theta_j\phi_j(\cdot)\quad\text{with}\quad\sum_{j=1}^\infty\theta_j^2 <\infty.
\end{equation}
Here $\mu$ can be either the Lebesgue measure in the context of density estimation or known probability measures on $[0,1]$ in the regression settings. 

Restrictions on the Fourier coefficients $\{\theta_j\}_{j=1}^\infty$ lead to several familiar classes such as the \textit{Sobolev ellipsoids} (e.g. Chapter 1.7.1 in \cite{tsybakov2008intro}) and \textit{full approximation set} (e.g. \cite{lorentz66}). Motivated by recent literature on high dimensional models (e.g. for a comprehensive review, see the handbook chapter by \cite{BCCHKN}), we introduce a new set of restrictions on the Fourier coefficients, and we call the resulting function class the \textit{approximate sparsity class}:
\begin{definition}\label{def:as}
For given constants $A>0, k>1/2$ and $C>0$, and for a given orthonormal basis $\Phi = \{\phi_j\}_{j=1}^\infty$, the approximate sparsity class is defined as
\begin{equation}
\begin{split}
\Theta_k(\Phi,A,C) := &\left\{f\in L^2([0,1],\mu): f(\cdot)=\sum_{j=1}^\infty\theta_j\phi_j(\cdot);\right.\\ &\left.\quad\text{the non-increasing re-ordering }\{\theta_{(j)}\}_{j=1}^\infty ~\text{satisfies}~|\theta_{(j)}|\leq Aj^{-k};\right.\\
&\left.\quad\forall J\geq 1,~\text{the tail sum satisfies}~\sum_{j=J+1}^\infty \theta_j^2\leq CJ^{-2k+1} \right\}.
\end{split}
\end{equation}
\end{definition}
First, we note that the re-ordering in the definition refers to the re-ordering of the coefficients by the magnitude of their absolute values. Specifically, after re-ordering, $\theta_{(j)}$ will be the j-th largest element in $\{\theta_j\}_{j=1}^\infty$ by absolute magnitude. We also want to remark that the re-ordering requirement in the definition is equivalent to that $\forall J\geq 1$, the non-increasing re-ordering of $\{\theta_{(j)}\}_{j=1}^J$ satisfies $|\theta_{(j)}|\leq Aj^{-k}$, which is a convenient characterization that we will use to establish various results in this paper. As we discussed in the introduction, it can be shown that the rate of decay of individual series coefficient is closely related to the smoothness of functions. The re-ordering condition relaxes such restriction by imposing less a priori knowledge on which series coefficients are important (as measured by magnitude).

Moreover, the restrictions on the tail sum is a natural one. In particular, if the Fourier coefficients decay at $Aj^{-k}$ without the re-ordering, the tail sum $\sum_{j=J+1}^\infty\theta_j^2$ can be shown to be bounded by $CJ^{-2k+1}$ using integral bound. However, by allowing the re-ordering of the Fourier coefficients, the tail-sum restriction is a necessary one. In particular, the tail-sum restriction imposes structures on the re-ordering, which also helps us bound the estimation bias.

At first sight, the complexity of the approximate sparsity space seems difficult to characterize. We will introduce additional spaces that have simpler structures and ``sandwich'' the approximate sparsity space.
\begin{definition}\label{def:sandwich}
For given constants $A>0,k>1/2$, and $C>0$, and for a given orthonormal basis $\Phi = \{\phi_j\}_{j=1}^\infty$, define the following spaces
\begin{equation}
\mathcal{E}_k(\Phi,A) := \left\{f\in L^2([0,1],\mu):\ f(\cdot)=\sum_{j=1}^\infty\theta_j\phi_j(\cdot);\ |\theta_j|\leq Aj^{-k}~\forall j\geq 1\right\}
\end{equation}
\begin{equation}
\begin{split}
\mathcal{A}_k(\Phi,A,C) :=& \left\{f\in L^2([0,1],\mu):f(\cdot)=\sum_{j=1}^\infty\theta_j\phi_j(\cdot);\right.\\&\quad\left.|\theta_1|\leq A;\ \forall J\geq 1,\sum_{j=J+1}^\infty \theta_j^2\leq CJ^{-2k+1}\right\}.
\end{split}
\end{equation}

\end{definition}

The class $\mathcal{E}_k(\Phi,A)$ consists of all functions in $L^2([0,1],\mu)$ whose Fourier coefficients decay in an ordered manner in polynomial rates. For example, differentiable functions in $L^2([0,1],\mu)$ can be viewed as elements in such class. On the other hand, the set $\mathcal{A}_k(\Phi,A,C)$ is an example of the full-approximation set discussed in \cite{lorentz66}. A function $f\in \mathcal{A}_k(\Phi,A,C)$ can be approximated by the partial sums $\sum_{j=1}^J\theta_j\phi_j(\cdot)$, and the tail sum restriction in the definition of $\mathcal{A}_k(\Phi,A,C)$ can be understood as the restriction on the bias from such approximation. We show in the appendix (Lemma \ref{lm:inclusion}) that for appropriately chosen constant $C$, $\mathcal{E}_k(\Phi,A)\subseteq\Theta_k(\Phi,A,C)\subseteq\mathcal{A}_k(\Phi,A,C)$. In particular, note that if $|\theta_j|<Aj^{-k}$, then
\begin{equation}
\sum_{j= J+1}^\infty\theta_j^2 \leq A^2\int_J^\infty t^{-2k}dt = \frac{A^2}{2k-1}J^{-2k+1}
\end{equation}
so we can simply take $C = A^2/(2k-1)$ in the definition of $\mathcal{A}_k(\Phi,A,C)$. The structures of $\mathcal{E}_k(\Phi,A)$ and $\mathcal{A}_k(\Phi,A,C)$ are much simpler and they will help us bound the metric entropy of the approximate sparsity class.

Throughout, we will use $M_2(\epsilon,\mathcal{F}):= \log N(\epsilon,\|\cdot\|_{L^2([0,1],\mu)},\mathcal{F})$ to denote the Kolmogorov $\epsilon$-entropy, where $N(\epsilon,\|\cdot\|_{L^2([0,1],\mu)}, \mathcal{F})$ is the cardinality of the largest $\epsilon$-packing set of a set $\mathcal{F}$ under the $L^2([0,1],\mu)$ distance. Next, we formally introduce the definition of minimax rates, borrowing notation from \cite{tsybakov2008intro}. Let $\{a_n\}$ and $\{b_n>0\}$ be real sequences. We write $a_n\gtrsim b_n$ if
$\liminf_{n\to \infty} a_n/b_n >0$, and similarly, we write $a_n\lesssim b_n$ if $\limsup_{n\to \infty} a_n/b_n <\infty$. We use the notation ``$\asymp$'' as the asymptotic order symbol. That is, we write $a_n\asymp b_n$ if 
\begin{equation}
0<\liminf_{n\to \infty} \frac{a_n}{b_n} \leq \limsup_{n\to \infty} \frac{a_n}{b_n} <\infty.
\end{equation}

\begin{definition}\label{def:minimax}
Given a set $\mathcal{F}\subseteq L^2([0,1],\mu)$ equipped with norm $\|\cdot\|_{L^2([0,1],\mu)}$, we say $\psi_n$ is a minimax optimal rate of convergence on $(\mathcal{F}, \|\cdot\|_{L^2([0,1],\mu)})$ if
\begin{equation}
\inf_{\hat{f}_n}\sup_{f\in\mathcal{F}}E_f\left[\|f-\hat{f}_n\|_{L^2([0,1],\mu)}^2\right]\asymp \psi_n^2
\end{equation}
where the infimum is taken over all possible estimators.
\end{definition}

\section{Entropy and Minimax Rate for Approximate Sparsity Class}

As we have discussed in the previous section, for appropriately chosen constant $C$ in the definitions of $\Theta_k(\Phi,A,C)$ and $\mathcal{A}_k(\Phi,A,C)$, we have $\mathcal{E}_k(\Phi,A)\subseteq\Theta_k(\Phi,A,C)\subseteq\mathcal{A}_k(\Phi,A,C)$. Therefore, understanding the metric entropy of $\mathcal{E}_k(\Phi,A)$ and $\mathcal{A}_k(\Phi,A,C)$ can help us control the metric entropy of the approximate sparsity class $\Theta_k(\Phi,A,C)$. In particular, the entropy on $\mathcal{E}_k(\Phi,A)$ and $\mathcal{A}_k(\Phi,A,C)$ will respectively give lower and upper bounds on the entropy of the sandwiched set $\Theta_k(\Phi,A,C)$, as shown in the next lemma.

\begin{lemma}\label{lm:akek}
The $\epsilon$-entropy of $\mathcal{E}_k(\Phi,A))$ under the $L^2([0,1],\mu)$ distance satisfies
\begin{equation}
M_2(\epsilon,\mathcal{E}_k(\Phi,A)) \gtrsim \epsilon^{-2/(2k-1)}.
\end{equation}
Moreover, $\epsilon$-entropy of $\mathcal{A}_k(\Phi,A,C))$ satisfies
\begin{equation}
M_2(\epsilon,\mathcal{A}_k(\Phi,A,C)) \asymp \epsilon^{-2/(2k-1)}.
\end{equation}
\end{lemma}

The proof of this lemma is given in the appendix. As we remarked before, $\mathcal{A}_k(\Phi,A,C)$ is a special case of the full approximation set introduced in \cite{lorentz66}, and its entropy can be established using Theorem 3 in \cite{lorentz66}. On the other hand, to the best of our knowledge, the earliest reference of the class $\mathcal{E}_k(\Phi,A)$ can be traced back to \cite{Smolyak60} in which the trigonometric basis $\Phi$ is considered \footnote{The class $\mathcal{E}_k(\Phi,A)$ is sometimes referred to as the ``hyperrectangles'' in the literature, and minimax risks over these hyperrectangles has been studied in the context of Gaussian shift models, see, for example, \cite{donoho90} and the references therein.}. Inspired by \cite{Smolyak60}, we adapt their arguments and extend similar results to $\mathcal{E}_k(\Phi,A)$ for more general orthonormal bases \footnote{Prior to learning the existence of \cite{Smolyak60}, we pursued an alternative route when establishing the entropy of $\mathcal{E}_k(\Phi,A)$. This alternative proof relies on an isometry between $\mathcal{E}_k(\Phi,A)$ and a generalized Hilbert cube (\cite{kloeckner12}). \cite{kloeckner12} establishes the bounds on the entropy of the Hilbert cubes with which one can infer lower bounds on the entropy of $\mathcal{E}_k(\Phi,A)$. While these lower bounds are sufficient for establishing the minimax rates, we opt to adopt the arguments used by \cite{Smolyak60} for the sake of clarity.}. Next, we use Lemma \ref{lm:akek} to establish the following theorem, which can be shown using a sandwich type of argument.

\begin{theorem}\label{thm:entropy}
The $\epsilon$-entropy of $\Theta_k(\Phi,A,C)$ under the $L^2([0,1],\mu)$ distance satisfies
\begin{equation}
 M_2(\epsilon,\Theta_k(\Phi,A,C)) \asymp \epsilon^{-2/(2k-1)}.
\end{equation}
\end{theorem}

This theorem formally establishes bounds on the $L^2$ metric entropy of our approximate sparsity class $\Theta_k(\Phi,A,C))$with the help of $\mathcal{E}_k(\Phi,A)$ and $\mathcal{A}_k(\Phi,A,C)$. The following corollary establishes a similar entropy result on the density subset of $\Theta_k(\Phi,A,C))$.

\begin{corollary}\label{col:entropy}
Let $\tilde\Theta_k(\Phi,A,C)\subseteq \Theta_k(\Phi,A,C)$ be defined as the subset of all probability densities in $\Theta_k(\Phi,A,C)$. Moreover, assume that the basis $\Phi =\{\phi_j\}_{j=1}^\infty$ includes a constant term and $\Theta_k(\Phi,A,C)$ is uniformly bounded. Then
\begin{equation}
M_2(\epsilon,\tilde\Theta_k(\Phi,A,C)) \asymp \epsilon^{-2/(2k-1)}.
\end{equation}
\end{corollary}

We note that many familiar orthonormal bases on $L^2([0,1],\mu)$ contain a constant term. Moreover, the requirement that the functions in $\Theta_k(\Phi,A,C)$ are uniformly bounded, can be satisfied, for example, when the orthonormal basis $\Phi$ is uniformly bounded and $k>1$. Assuming uniform boundedness will allow us to find a set of the densities of the form $(f+M'+1)/(\int f d\mu + M'+1)$ for $f\in\mathcal{E}_k(\Phi,A)$ for some large constant $M'$, which provides a lower bound for the entropy of $\tilde{\Theta}_k(\Phi,A',C')$ for some constants $A'$ and $C'$, which in turn is a lower bound for $\tilde{\Theta}_k(\Phi,A,C)$. Then applying the sandwich argument again gives us the results of the corollary. With the entropy results, we formally establish the minimax rates for nonparametric regression with Gaussian noise and density estimation on the approximate sparsity classes.

\begin{theorem}\label{thm:minimax}
Let $\tilde\Theta_k(\Phi,A,C)\subseteq \Theta_k(\Phi,A,C)$ be the subset of all probability densities in $\Theta_k(\Phi,A,C)$ and assume $\Theta_k(\Phi,A,C)$ is uniformly bounded.
\begin{itemize}
    \item[(i)] Suppose $\{X_i,Y_i\}_{i=1}^n$ is i.i.d with $X_i\sim P_X$ that admits a density. For nonparametric regression with Gaussian noise model $Y = f(X)+ \epsilon$ with $\epsilon\sim N(0,\sigma^2)$ and $f\in\Theta_k(\Phi,A,C)$, the minimax rate of convergence satisfies
    \begin{equation}
    \inf_{\hat{f}_n}\sup_{f\in\Theta_k(\Phi,A,C)}E\left[\|\hat{f}_n-f\|^2_{L_2([0,1],P_X)}\right] \asymp n^{-\frac{2k-1}{2k}}.
    \end{equation}
    \item[(ii)] Suppose $\{X_i\}_{i=1}^n$ is i.i.d with $X_i\sim X$ for $X\in[0,1]\subseteq\mathbf{R}$ with density $f\in \tilde\Theta_k(\Phi,A,C)$, the minimax rate of convergence satisfies
    \begin{equation}
    \inf_{\hat{f}_n}\sup_{f\in\tilde\Theta_k(\Phi,A,C)}E_f\left[\|\hat{f}_n -f\|^2_{L_2([0,1],\mu)}\right] \asymp n^{-\frac{2k-1}{2k}}.
    \end{equation}
\end{itemize}

\end{theorem}

Theorem \ref{thm:minimax} can be seen as a direct consequence of Theorem \ref{thm:entropy} and Corollary \ref{col:entropy} due to the tight connection between the entropy and minimax rates established in \cite{yangbarron99} and the references therein. In particular, the minimax rate of estimation on a particular class of functions $\mathcal{F}$ is closely related to the ``critical separation'' $\epsilon_n$, which is determined via the identity $M_2(\epsilon_n,\mathcal{F}) = n\epsilon_n^2$. Then the results of the theorem should follow from our entropy results. We note that there are some caveats when applying \cite{yangbarron99} directly on the density subset $\tilde\Theta_k(\Phi,A,C)$ since it is not convex nor bounded away from zero. However, we show in the proof that this subset is sandwiched in between two convex sets of densities with which we can establish the desired result.

\section{Nearly Minimax Optimal Adaptive Density Estimator}

\subsection{Preliminaries}
In this section, we propose a density estimator that is adaptive and achieves the minimax rate of convergence on $\tilde\Theta_k(\Phi,A,C)$ up to a log term. Suppose $X$ is a random variable with probability density $f\in L^2([0,1],\mu)$ and let $\Phi = \{\phi_j\}_{j=1}^\infty$ be an orthonormal basis of $L^2([0,1],\mu)$ with $\mu$ being the Lebesgue measure. Then $f$ has a unique representation:
\[
f(\cdot) = \sum_{j=1}^\infty\theta_j\phi_j(\cdot).
\] 
Moreover, $f$ being a probability density and $\Phi = \{\phi_j\}_{j=1}^\infty$ being an orthonormal basis allow us to find expressions of $\theta_j$'s in terms of expectations:
\begin{equation}\label{series:id}
\theta_j = \int_{[0,1]}\phi_j(x)(\sum_{j=1}^\infty\theta_j\phi_j(x))d\mu(x) =\int_{[0,1]}\phi_j(x)f(x)d\mu(x) = E\left[\phi_j(X)\right].
\end{equation}
Given an i.i.d sample $\{X_i\}_{i=1}^n$ with $X_i\sim X$, the standard series estimator for the density $f$ builds on identity (\ref{series:id}) and takes the form
\[
\hat{f}_{J}(x) = \sum_{j=1}^J\hat\theta_j\phi_j(x)\quad\text{where}\quad\hat\theta_j = \frac{1}{n}\sum_{i=1}^n\phi_j(X_i)
\]
Note that the mean integrated squared error (MISE) of $\hat{f}_n$ equals
\begin{equation}\label{series:mise}
E\left[\|f-\hat{f}_{J}\|_{L^2([0,1],\mu)}\right] = \sum_{j=1}^J\frac{Var(\phi_j(X))}{n} + \sum_{j=J+1}^\infty\theta_j^2.
\end{equation}
The choice of the cutoff $J$ plays an essential role in the trade-off between variance and bias in (\ref{series:mise}), and yet in order to properly choose $J$ in estimation, one often has to assume some prior knowledge about the smoothness of the unknown $f$.

We propose an alternative estimator that circumvents this issue. In particular, our estimator requires choosing a large cutoff $J$, and then uses LASSO to select the most ``relevant'' Fourier coefficients among $\{\hat\theta_1,\hat\theta_2,\cdots,\hat\theta_J\}$. Note that many of the familiar bases, such as cosine and Legendre polynomial basis, contain a constant element. Therefore, we always assume the orthonormal basis of choice contains a constant term, which, without loss of generality, we assume to be the first basis term $\phi_1$. Then note that $\hat\theta_1 = n^{-1}\sum_{i=1}^n\phi_1 = \phi_1 = E\left[\phi_1\right] = \theta_1$, that is, there's no estimation error for the first term. Therefore, we should always include $\hat\theta_1$ in our estimator in practice.

For a slight change of notation, let $\theta^J = \{\theta_1,\theta_2,\cdots,\theta_J\}\in\mathbf{R}^{J}$ denote the first $J$ true but unknown series coefficients. Let $\hat\theta^J\in\mathbf{R}^{J}$ be an estimator of $\theta^J$ defined as follows
\begin{equation}\label{theta:naive}
\hat\theta_j = \frac{1}{n}\sum_{i=1}^n\phi_j(X_i);\quad\hat\theta^J:=(\hat\theta_1,\cdots,\hat\theta_J)
\end{equation}
where each $\hat\theta_j$ is consistent by weak law of large numbers under mild regularity conditions, and $\hat\theta^J$ can be shown to be consistent using Bernstein's inequality and maximal inequality under additional assumptions (e.g. if the orthonormal basis $\Phi$ is bounded or if the density is bounded). Given $\hat\theta^J$, the estimator we consider here is the so-called \textit{hard-thresholding} estimator:
\begin{equation}\label{theta:hard}
\tilde\theta_j = \omega_j\hat\theta_j\quad\text{where}\quad \omega_j = \mathbf{1}\{|\hat\theta_j|\geq\lambda\} \quad\text{for}\ 1\leq j\leq J
\end{equation}
where $\omega_j$'s are the thresholding parameters depending on the penalty parameter $\lambda$. To be consistent with the notation in (\ref{threshold}), we let $\omega_j = 0$ for $j>J$. The hard thresholding estimator we consider here differs from the \textit{soft-thresholding} estimator that has been considered in some previous literature, for example, \cite{bunea2010spades} and the references therein. Intuitively, for a properly chosen penalty parameter $\lambda$, we penalize ``small'' estimates of the series coefficients to $0$ while keeping the rest unchanged. Let $T\subseteq \{1,\cdots,J\}$ denote the set of selected indices, that is,
\begin{equation}\label{def:T}
T := \{j\in\{1,\cdots, J\}: |\hat{\theta}_j|\geq\lambda\}.
\end{equation}
As a result, the new estimator we get is
\begin{equation}\label{series:post}
\tilde{f}_J(x) := \sum_{j=1}^J\tilde\theta_j\phi_j(x) = \sum_{j=1}^\infty\omega_j\hat\theta_j\phi_j(x) = \sum_{j\in T}\hat\theta_j\phi_j(x).
\end{equation}
As we will show later, the quality of this estimator depends on the penalty parameter and, to a lesser degree, the cutoff $J$.

Moreover, note that $\tilde{f}_J$ is not necessarily a probability density. We need to make sure $\tilde{f}_J$ integrates to $1$ and is nonnegative. The former is easily satisfied if $\mu$ is the Lebesgue measure on $[0,1]$ and the first basis element $\phi_1 = 1$, which, by the definition of orthonormality, implies
\begin{eqnarray}
\int_{[0,1]}\tilde{f}_J(x)d\mu(x) &=& \int_{[0,1]} \phi_1^2 d\mu(x) + \sum_{j\in T\setminus\{1\}}\hat\theta_j\phi_j(x) d\mu(x)\nonumber \\ &=& \int_{[0,1]} \phi_1^2 d\mu(x) + \frac{1}{\phi_1}\int_{[0,1]}\phi_1\sum_{j\in T\setminus\{1\}}\hat\theta_j\phi_j(x) d\mu(x) = 1.
\end{eqnarray}
This is another reason why we should always include $\hat\theta_1 = \phi_1$ in our estimator.

On the other hand, some forms of post-processing are needed to ensure that the estimator is nonnegative. We will follow the \textit{P-Algorithm} suggested by \cite{Gajek86}, which is attractive in both its ease of implementation and the statistical properties of the resulting estimator. Here we review the P-Algorithm and illustrate with our estimator:

\begin{definition}\label{p-alg}
The P-Algorithm is defined as follows:

\begin{itemize}
    \item[1.] Set $f_0 = \tilde{f}_J$, and $k=0$.
    \item[2.] Set $f_{k+1} = \max\{0,f_k\}$, and let $C_{k+1} = \int_{[0,1]}f_{k+1}(x)d\mu(x)$. Stop if $C_{k+1} = 1$.
    \item[3.] Set $f_{k+2} = f_{k+1} - (C_{k+1} -1)$.
    \item[4.] Set $k = k+2$ and go to step 2.
\end{itemize}

Denote the resulting estimator from the algorithm as $\hat{f}^*$.
\end{definition}

\cite{Gajek86} shows that the P-Algorithm converges both point-wise and in $\|\cdot\|_{L^2([0,1],\mu)}$ to a density $\hat{f}^* = \max\{0,\tilde{f}_J + c\}$, for some constant $c$. Moreover, $\hat{f}^*$ is the orthogonal projection of $\tilde{f}_J$ onto the space of densities, and $\hat{f}^*$ has at least the same rate of convergence (as measured in MISE) as $\tilde{f}_J$.

\subsection{Main results}

To estimate $\hat{f}^*$ and to establish its statistical properties, we need to consider a data-driven way of choosing the regularization parameter $\lambda$ in (\ref{theta:hard}). In particular, in order to control the penalization error in a uniform manner, we want to pick regularization parameter $\lambda$ in a way such that
\begin{equation}\label{lambda}
\lambda\quad\geq\quad (1-\alpha)-\text{quantile of}\ \|\hat\theta^J-\theta^J\|_{\infty}.
\end{equation}
There are several ways of achieving this. For example, one can approximate such $\lambda$ using the results from the self-normalized moderate deviation theory or high dimensional bootstrap literature, see for example \cite{BCCHKN} and the references therein. However, in order to establish the MISE, we need (\ref{lambda}) to hold with $\alpha = \alpha_n$ going zero sufficiently fast as sample size $n$ increases. To this end, we will modify the $\lambda$ proposed by \cite{BCCHKN} based on the moderation deviation theories, and instead we will use the Talagrand's inequality (see \cite{Bousquet2003Talagrand}) to establish the asymptotic result.

To facilitate the discussion, we borrow the following notations from \cite{BCCHKN}: Let $\{X_i\}_{i=1}^n$ be an i.i.d random sample in $\mathbf{R}$ and let $\{\phi_j\}_{j=1}^\infty$ be an orthonormal basis of $L^2([0,1],\mu)$. With some abuse of notation, let $\Phi$ denote the CDF of the standard normal variable, and define
\begin{equation}\label{def:Z}
Z_{ij}:= \phi_j(X_i) - E\left[\phi_j(X_i)\right]
\end{equation} 
and we let $\hat{Z}_{ij}$ be the sample analog of $Z_{ij}$
\begin{equation}\label{def:Zhat}
\hat Z_{ij} = \phi_j(X_i) - \frac{1}{n}\sum_{k=1}^n\phi_j(X_k).
\end{equation} 
We introduce the following regularity conditions.

\begin{condition}\label{regularity}\
Let $Z_{ij}$ be defined as in (\ref{def:Z}) and suppose that the orthonormal basis $\{\phi_j\}_{j=1}^\infty$ satisfies $\max_{1\leq j\leq J}\|\phi_j(\cdot)\|_\infty \leq M_J$ for some $M_J$.

\begin{itemize}
    \item[(i)] $n^{-1}\sum_{i=1}^n E\left[Z_{ij}^2\right]\geq 1$ for all $1\leq j\leq J$;
    \item[(ii)] $J = n^p$ for some known $p>0$;
    \item[(iii)] $M_J^2\leq n/\log(n)$; 
    \item[(iv)] $(JM_J/n)\alpha_n^{1/2}\leq n^{-(2k-1)/(2k)}$ for some positive sequence $(\alpha_n)\downarrow 0$.
\end{itemize}

\end{condition}

Condition \ref{regularity} will be assumed to establish the main results in this section. We want to emphasize that all but (i) in Condition \ref{regularity} can be verified, and we note that (i) is stated in the form that is convenient for the proof and the lower bound $1$ in the statement can be replaced by any positive constant. In particular, (ii) is chosen by researchers, and (iii) can be checked for a given orthonormal basis. Note that when the chosen orthonormal basis is uniformly bounded, such as the cosine basis, we have $M_J^2 = M$ for some constant $M$, in which case (iii) is trivially satisfied, and we can potentially choose $J\gg n$. On the other hand, if we have a growing basis, the choices of $J$ can be limited depending on how fast the basis grows. For example, the Legendre basis grows with $M_J^2 = 2J+1$, and we have to choose $J$ such that $J = n^p$ for some $p<1$ ($p$ can be close to $1$ for $n$ large). The requirement (iv) can be viewed as the extra cost we pay when bounding the variance term in the MISE. In the next theorem, we propose a $\lambda$ such that $\lambda\geq (1-\alpha_n)$-quantile of $\|\hat{\theta}^J-\theta^J\|_\infty$ with $\alpha_n$ goes to zero sufficiently fast so that the MISE can be established.

\begin{theorem}\label{thm:ta}
Let $\lambda$ be defined as follows,
\begin{equation}\label{ta:lambda}
\lambda :=\sqrt{\frac{\log(J)}{n}}\Phi^{-1}\left(1-\frac{1}{2\sqrt{2\pi}}\frac{1}{\sqrt{2\log(J)}}\frac{1}{J}\right)\max_{1\leq j\leq J}\left(\frac{1}{n}\sum_{i=1}^n \hat{Z}_{ij}^2\right)^\frac{1}{2}.
\end{equation}
Assume (i)-(iii) in Condition \ref{regularity} are satisfied. Then there exists $n^*\in\mathbf{N}$ such that for all $n\geq n^*$,
\[
\lambda\geq (1-\alpha_n)-\text{quantile of}\ \|\hat{\theta}^J-\theta^J\|_\infty
\]
with $\alpha_n = (Jn)^{-2} + 2n^{-3}$.
\end{theorem}

The proof of the theorem is given in the appendix and it can be easily adapted to allow other choices of $\alpha_n$ so that (iv) in Condition \ref{regularity} can be satisfied. The expression of $\lambda$ in the theorem bears some similarities to the ones given in \cite{BCCH} and in \cite{BCCHKN} (Theorem 2.4) using moderate deviation theories. Compared to their constructions, for $J = n^p$, our $\lambda$ has an extra $\sqrt{\log(n)}$ multiplied (ignoring the constants). Roughly speaking, by the Talagrand's inequality, this extra $\log$ term pushes our $\lambda$ into the exponential tail and allows us to obtain $\alpha_n$ that is sufficiently fast for establishing MISE. We remark that the $\lambda$ based on the moderate deviation theories can still be used to establish the non-asymptotic rate. 

With this result, we can establish the second main result of this section. In particular, we are going to assume that the density of the random variable in question belongs to the approximate sparsity space $\tilde\Theta_k(\Phi,A,C)$, and we will show that the post-processed estimator $\hat{f}^*$ defined in \ref{p-alg} admits the minimax rate of convergence on $\tilde\Theta_k(\Phi,A,C)$ up to a log term.

\begin{theorem}\label{thm:rates}
Let $\{X_i\}_{i=1}^n\sim X$ be an i.i.d random sample with support $[0,1]\subset \mathbf{R}$. Assume that the true density $f$ of $X$ is in $\tilde\Theta_k(\Phi,A,C)$ uniformly bounded by some constant $\tilde{C}$ and that Condition \ref{regularity} is satisfied. Let the regularization parameter $\lambda$ be chosen as in (\ref{ta:lambda}) and let $\hat{f}^*$ be the estimator defined in \ref{p-alg}. Then 
\[
\sup_{f\in\tilde\Theta_k(\Phi,A,C)} E_f[\|f-\hat{f}^*\|_{L_2}^2] = O\left(\left(\frac{\log^2(n)}{n}\right)^{\frac{2k-1}{2k}}\right).
\]
\end{theorem}

The proof of the theorem is given in the appendix. To bound the MISE, we first decompose the MISE into roughly the standard ``variance'' and ``bias'' components, and then we bound each separately with the help of Theorem \ref{thm:ta}. In the assumptions of the theorem, $f\in \tilde\Theta_k(\Phi,A,C)$ imposes restrictions on the Fourier coefficients of $f$, which will play a crucial role in establishing the rates. Moreover, we assume that $\tilde\Theta_k(\Phi,A,C)$ is uniformly bounded, which is also assumed when we establish the minimax rates. Note that for the cases when $k>1$ and the orthonormal basis $\Phi$ is uniformly bounded, one can verify that $\tilde\Theta_k(\Phi,A,C)$ is uniformly bounded. In addition, Condition \ref{regularity} is assumed so that we can utilize the results in Theorem \ref{thm:ta}, and as we commented before, the requirements in Condition \ref{regularity} are easy to verify. 

We also want to emphasize the adaptive nature of our estimator. In particular, once the researchers have decided on which orthonormal basis to use, the construction of the estimator $\hat{f}^*$ with data-driven $\lambda$ does not depend on the assumptions on the approximate sparsity class (e.g. how fast the Fourier coefficients decay). Theorem \ref{thm:rates} simply states that our estimator achieves the minimax rate on any approximate sparsity class. This is attractive in practice since the researchers do not have to assume the smoothness of the true data-generating density other than that it belongs to \textit{some} approximate sparsity class.

\subsection{Example Using Cosine Basis}
In this section, we are going to illustrate our previous results with the differentiable densities and the cosine basis. The standard cosine orthonormal basis $\{\phi_j\}_{j=1}^\infty$ on $L^2([0,1],\mu)$ is defined as follows:
\begin{equation}\label{cos_basis}
\phi_1(x) = 1;\quad \phi_j(x)=\sqrt{2}\cos(\pi (j-1)x),~j=2,3,\cdots
\end{equation}
which is uniformly bounded. Suppose $f\in L^2([0,1],\mu)$ is twice differentiable. Then for $\{\phi_j\}_{j=1}^\infty$ being the cosine basis defined above, $f(\cdot) = \sum_{j=1}^\infty\theta_j\phi_j(\cdot)$ and there exists some constant $c$ such that 
\[
|\theta_j|\leq cj^{-2}\int_0^1|f^{(2)}(x)|~dx,\quad j\geq 1
\]
As noted in \cite{efro2008} (Section 2.2), unless additional (boundary) assumptions are made on the function $f$, the series coefficients can not decay faster than the rate $j^{-2}$ (regardless of the smoothness of functions). This makes the set of twice differentiable functions with bounded second derivative a special example of the approximate sparsity class. In fact, this is an example of $\mathcal{E}_k(\Phi,A)$, which are themselves a special case of the approximate sparsity class without reordering. The next result follows directly from Theorem \ref{thm:rates}.

\begin{corollary}\label{col:cos_rates}
Suppose $\{X_i\}_{i=1}^n\sim X$ is an i.i.d sample of random variables with $X_i\in [0,1]$. Suppose the true density $f$ of $X$ is such that $f\in L^2([0,1],\mu)$ and $f$ is twice differentiable with bounded second derivative. Let $\{\phi_j\}_{j=1}^{\infty}$ be the cosine orthonormal basis defined as in (\ref{cos_basis}). Assume requirement (i) in Condition \ref{regularity} is satisfied and let $J = n^p$ for some $p>0$. Let the regularization parameter $\lambda$ be chosen as in (\ref{ta:lambda}) and let $\hat{f}^*$ be the estimator defined in \ref{p-alg}. Then 
\[
E\left[\|f-\hat{f}^*\|_{L_2}^2\right] = O\left(\left(\frac{\log^2(n)}{n}\right)^{\frac{3}{4}}\right).
\]
\end{corollary}

Note that since the cosine basis is uniformly bounded, we can drop the assumption that the density is bounded and the requirements in Condition \ref{regularity} can be verified. Moreover, in view of our results in Section 3, the rate in this corollary is minimax up to a log term. As the approximate sparsity class and $\mathcal{E}_k(\Phi,A)$ class are new (and we commented previously that the familiar Sobolev ellipsoids are subsets of these classes), we have not yet considered the performances of other density estimators on such classes, and whether there are other adaptive rate-optimal estimators on these new classes remains an open question.

\section{Simulations}

In this section, we conduct simulation studies in which we compare the simulated MISE of our estimator introduced in Section 4 with an alternative estimator for which the series cutoff has to be properly specified. The true data generating density is constructed to be in the approximate sparsity class. 


\subsection{Design density in approximate sparsity class}
For the first simulation, we construct a density whose (cosine) Fourier coefficients satisfy approximate sparsity. Let 
\begin{equation}\label{as_design}
f_X(\cdot) = \sum_{j=1}^\infty\theta_j\phi_j(\cdot)
\end{equation}
where $\phi_j$'s are the cosine basis terms defined in (\ref{cos_basis}). We specify the $\theta_j$'s in the following way:
\begin{itemize}
    \item $\theta_1 = 1$, $\theta_j = Aj^{-2}$ for $j = 2,3,7,8$;
    \item $\theta_5 = A10^{-2}$, $\theta_{11} = A4^{-2}$, $\theta_{13} = A6^{-2}$, $\theta_{14} = A5^{-2}$, $\theta_{15} = A9^{-2}$;
    \item $\theta_j = 0$ for $j = 4,6,9,10,12$ and for all $j\geq 16$.   
\end{itemize}
As we've shown in section 4, since $\theta_1 = 1$, $f_X$ integrates to 1. The constant $A = 2$ is chosen such that $f_X$ in (\ref{as_design}) is non-negative and hence a proper probability density.

\subsection{Simulation procedures}

To simulate MISE, we proceed in following steps:

\emph{Step 1}: We draw $B = 1000$ independent i.i.d. samples $\{X_i\}_{i=1}^{N}$ of size $N$. Here the $N$ we consider will be $N = 5000, 10000, 15000, 20000$.
    
\emph{Step 2(a)}: For each sample $\{X_i\}_{i=1}^{N}$ from the true density $f_X$, we construct an estimator $\hat{f}^*$ described in section 4:
    \begin{itemize}
        \item We use cosine orthonormal basis defined in (\ref{cos_basis});
        \item The pre-processed estimator $\tilde{f}_J$ is constructed as in (\ref{series:post}), with $J = 200$;
        \item The penalty parameter $\lambda$ is chosen as in (\ref{ta:lambda});
        \item We pass $\tilde{f}_J$ to the p-algorithm defined in (\ref{p-alg}):
        \begin{itemize}
            \item The integral $C_{k+1}:= \int_0^1 f_{k+1}(x) dx$ at each iteration is approximated using numerical integration, and we denote the approximated integral as $\hat{C}_{k+1}$;
            \item The p-algorithm stops when $|\hat{C}_{k+1} - 1| < e^*$ for user specified $e^*$. This returns the positive estimator $\hat{f}^* = f_{k+1}$.
        \end{itemize}
    \end{itemize}

\emph{Step 2(b)}: For each sample $\{X_i\}_{i=1}^{N}$ from the true density $f_X$, we construct an alternative comparison estimator $\check{f}$ in the following way:
    \begin{itemize}
        \item We use cosine orthonormal basis defined in (\ref{cos_basis});
        \item We pick series cutoff $J$ to be $N^{1/4}$ and construct the natural estimator
        \begin{equation}
        \hat{f}_J(\cdot) = \sum_{j=1}^J\hat\theta_j\phi_j(\cdot),\quad \text{where}\quad \hat\theta_j = \frac{1}{N}\sum_{j=1}^n \phi_j(X_i)
        \end{equation}
        \item We pass $\hat{f}_J$ to the p-algorithm defined in (\ref{p-alg}) in the same way as in Step 2(a), which returns a positive estimator $\check{f}$.
    \end{itemize}

\emph{Step 3}: For each $1\leq b\leq B$ sample defined in \emph{Step 1} and associated estimator from \emph{Step 2(a)}, we estimate the integrated squared error (ISE) using numerical integration, and we use $\widehat{ISE}_b$ to denote the approximated ISE. We then calculate the estimated MISE
\begin{equation}
\widehat{MISE} := \frac{1}{B}\sum_{b=1}^{B} \widehat{ISE}_b
\end{equation}
Repeat the same process for the estimator from Step 2(b).

\begin{remark}
Unlike our estimator in \emph{Step 2(a)}, for the comparison estimator in \emph{Step 2(b)} we need to properly specify the series cutoff $J$. In the simulation, for this comparison estimator, we choose $J = N^{1/4}$ for the following reason. When the true density is unknown, if the researcher is willing to make assumptions on the smoothness of the density, they can then determine the series cutoff based on such assumptions. For example, if the researcher assumes the true density is twice differentiable with bounded second derivative, as we discussed in Section 5, the Fourier coefficients $\theta_j$ decay at the rate $j^{-2}$, and one can show that a series cutoff $J = N^{1/4}$ minimizes the MISE under such assumption. 
\end{remark}

\subsection{Simulation Results}
The simulation design described above is coded directly using Python, with code available upon request. Since the design density is not a known density coded in Python packages, we use inverse transform sampling to generate random samples from our design density. We illustrate the performance of the inverse transform sampling in Figure (\ref{fig:random_sample}), where the solid line is the design density, and the normalized histogram below the solid line is constructed using a random sample of size $10000$ from such inverse transform sampling.

\begin{figure}[ht]
\includegraphics[width = \linewidth]{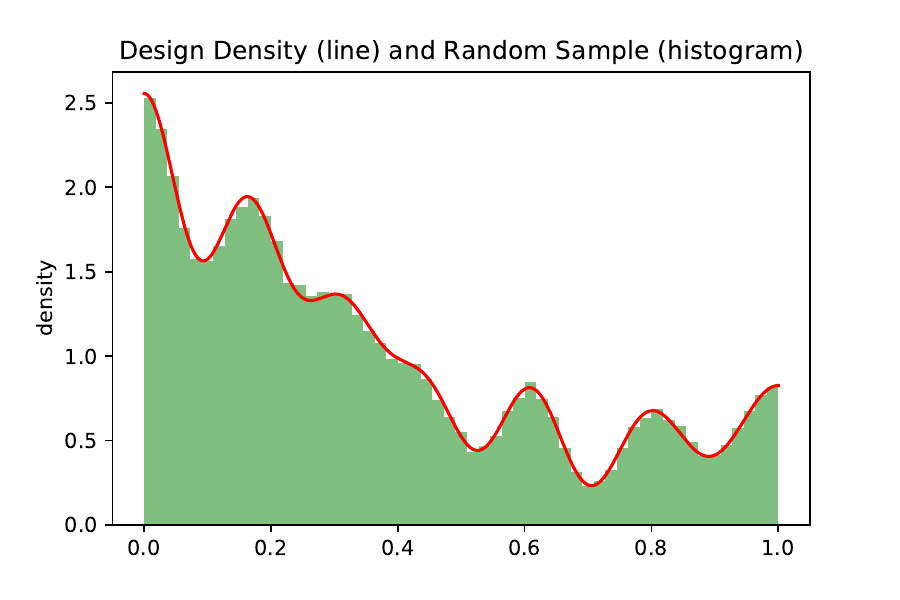}
\caption{Design Density with Random Sample}
\label{fig:random_sample}
\end{figure}

We present the simulation results in Figure (\ref{fig:MISE}). Each dot represents the simulated MISE for a given sample size, where the red dots (labeled as ``f\_star'') are the simulated MISEs for our estimator and the blue dots (labeled as ``f\_check'') are for the comparison estimator. We make the following observations. First, as expected, for both our estimator and the comparison estimator, the simulated MISE decreases as sample size increases. For the comparison estimator, this reflects the fact that such estimator is using larger series cutoffs as sample size increases, which will eventually include all the nonzero terms of the series coefficients. On the other hand, our estimator requires choosing a large series cutoff ($J=200$) to start with and then uses a data-driven hard-thresholding procedure to select the relevant terms. Second, for each of the sample sizes in our consideration, the simulated MISE of our estimator is smaller than that of the comparison estimator. This is likely due to the special feature of our design density, where the large Fourier coefficients show up in later series terms. Our estimator estimates first $J=200$ series terms and the data-driven hard-thresholding procedure is able to pick up the large series terms to reduce the bias. On the other hand, the comparison estimator has to specify the series cutoff $N^{1/4}$ and may fail to include the large Fourier coefficients that show up in later terms.

These simulation results demonstrate that our estimator performs well (as measured by the simulated MISE) for estimating densities in the approximate sparsity class. Moreover, the results showcase the adaptive nature of our estimator in comparison to an alternative estimator for which the researcher has to make potentially restrictive smoothness assumptions in order to determine the proper series cutoff. Even then, if the sample size is not large enough, such comparison estimator may still miss large Fourier coefficients that show up in later series terms, which can result in larger estimation errors.

\begin{figure}[ht]
\includegraphics[width = \linewidth]{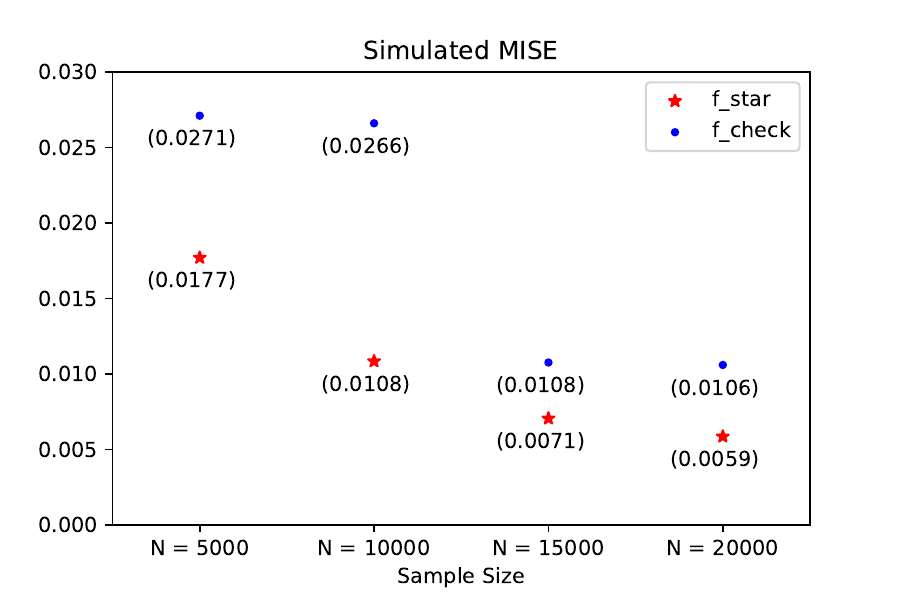}
\caption{Simulated MISE With Various Sample Sizes}
\label{fig:MISE}
\end{figure}

\section{Concluding Remarks}
In this paper, we have studied a new class of functions, which we call the approximate sparsity class, that is characterized by a new set of restrictions on the Fourier coefficients of those functions for a given orthonormal basis. We have derived the upper and lower bounds on $L^2$ $\epsilon$-metric entropy of the approximate sparsity class and we have established the minimax rates for nonparametric regression with Gaussian noise and for density estimation. We have shown that functions in such classes are natural candidates for the thresholding types of estimators and we proposed an adaptive density estimator that is nearly minimax optimal (up to a log term) over such class. For future research, we hope to study estimators for nonparametric regression for functions in the approximate sparsity class and we hope to generalize the approximate sparsity class to high-dimensional regression settings.

\begin{appendix}
\section{Proofs for Section 3}

\subsection{Proof of Lemma \protect\ref{lm:akek}}

First, we establish the metric entropy of $\mathcal{A}_k(\Phi,A,C)$, which relies on Theorem 3 in \cite{lorentz66}. In particular, note that in the definition of $\mathcal{A}_k(\Phi,A,C)$, we have $\sum_{j=J+1}^\infty\theta_j^2 \leq CJ^{-2k+1}$. So following \cite{lorentz66} notation, let $\delta_J^2  = CJ^{-2k+1}$. Moreover, simple calculation shows that $\delta_{2J}\leq c\delta_J$ for some constant $0<c<1$. Therefore, condition (13) in \cite{lorentz66} holds.

Let $\epsilon>0$ be given. Then by Theorem 3 of \cite{lorentz66}, the $\epsilon$-metric entropy of $\mathcal{A}_k(\Phi,A,C)$ is of order
\[
\min\{J:\delta_J^2 = C J^{-2k+1} \leq \epsilon^2\}
\]
solving which, we get $M_2(\epsilon,\mathcal{A}_k(\Phi,A,C))\asymp \epsilon^{-2/(2k-1)}$.

Second, we establish the metric entropy of $\mathcal{E}_k(\Phi,A)$. Recall that the set $\mathcal{E}_k(\Phi,A)$ is defined as
\[
\mathcal{E}_k(\Phi,A) := \left\{f\in L^2([0,1],\mu): f(\cdot)=\sum_{j=1}^\infty\theta_j\phi_j(\cdot);~|\theta_j|\leq Aj^{-k}~\forall j\geq 1\right\}.
\]
which is isometric to the set
\[
E_{k,A}:= \left\{ (\theta_j)_{j=1}^\infty: |\theta_j|<Aj^{-k} \right\}
\]
i.e. for any two elements $f,g\in \mathcal{E}_k(\Phi,A)$, we can find two unique elements $f^*, g^*$ in $E_{k,A}$ such that $\|f-g\|_{L^2([0,1],\mu)} = \|f^*-g^*\|_{\ell^2}$. This follows from the fact that $\Phi$ is an orthonormal basis. Define the following set
\[
E_{k,A}^Q := \left\{ (\theta_j)_{j=1}^\infty: |\theta_j|<Aj^{-k}\ \forall\ 1\leq j\leq Q;\ \theta_j = 0\ \forall j\geq Q+1 \right\}.
\]
which is a subset of $E_{k,A}$. It can be shown that $E_{k,A}^Q$ is isometric to the following set
\[
H_{k,A}^Q := \left\{ (\theta_j)_{j=1}^Q: |\theta_j|<Aj^{-k}\ \forall\ 1\leq j\leq Q \right\} \subseteq \mathbf{R}^Q.
\]
By the definition of isometry, the metric entropy of $\mathcal{E}_k(\Phi,A)$ is lower-bounded by the metric entropy of $H_{k,A}^Q$. We show that for a properly chosen $Q$, we can establish the desired lower bound.

Let $V$ denote the volume of $H_{k,A}^Q$ and let $v$ denote the volume of $\epsilon$-ball in $\mathbf{R}^Q$. Then the ratio $V/v$ provides a lower bound on the covering number of $H_{k,A}^Q$. By Sterling's formula, we have 
\[
v \asymp \frac{1}{Q\pi}\left(\frac{2\pi e}{Q}\right)^{\frac{Q}{2}}\epsilon^Q
\]
Then we have
\begin{align*}
\frac{V}{v} \asymp \left(\prod_{j=1}^Q 2Aj^{-k}\right)/\left(\frac{1}{Q\pi}\left(\frac{2\pi e}{Q}\right)^{\frac{Q}{2}}\epsilon^Q\right)
\end{align*}
Taking $\log$, we have
\begin{align*}
\log\left(\frac{V}{v}\right) &\gtrsim \sum_{j=1}^Q\log(2Aj^{-k}) - \log\left(\frac{1}{Q\pi}\left(\frac{2\pi e}{Q}\right)^{\frac{Q}{2}}\epsilon^Q \right)\\
&\gtrsim -k\sum_{j=1}^Q\log(j) + Q\log(2A)- \log(Q^{-(Q+1)/2}\epsilon^Q)\\
&\gtrsim -k(Q\log(Q)-Q + \Theta(\log(Q)))+ Q\log(2A)- \log(Q^{-(Q+1)/2}\epsilon^Q)\\
&\gtrsim (-k+\frac{1}{2})Q\log(Q) - Q\log(\epsilon) + (\log(2A) + k)Q
\end{align*}
where the third line holds by Sterling's approximation. Then take $Q$ to be such that $Q^{-k + 1/2} = \epsilon^{-1}$, in which case $Q = \epsilon^{-2/(2k-1)}$, and we have
\begin{align*}
\log\left(\frac{V}{v}\right) &\gtrsim \epsilon^{-2/(2k-1)}.
\end{align*}
This gives a lower bound on the capacity ($\log$ of covering number) of $H_{k,A}^Q$. Using the convenient fact that the packing number is bounded below by covering number (see for example, \cite{lorentz66}), we conclude that
\[
M_2(\epsilon,\mathcal{E}_k(\Phi,A)) \geq M_2(\epsilon,H_{k,A}^Q)\gtrsim \epsilon^{-2/(2k-1)}.
\]
\qed

\subsection{Proof of Theorem \protect\ref{thm:entropy}}
By Lemma \ref{lm:inclusion}, for properly chosen constant $C$ in the definition of $\mathcal{A}_k(\Phi,A,C)$, we have $\mathcal{E}_k(\Phi,A)\subseteq\Theta_k(\Phi,A,C)\subseteq\mathcal{A}_k(\Phi,A,C)$, and the results follow from Lemma \ref{lm:akek}. \qed

\subsection{Proof of Corollary \protect\ref{col:entropy}}
We will use the sandwich argument again to establish this result. Note we are given that $\Theta_k(\Phi,A,C)$ are uniformly bounded by some constant $M'$, which also holds for its subset $\mathcal{E}_k(\Phi,A)$. We will manipulate $\mathcal{E}_k(\Phi,A)$ to establishes the lower bound, while the upper bound is still given by the full approximation set $\mathcal{A}_k(\Phi,A,C)$.

First, consider the following transformation
\begin{align*}
\tilde{\mathcal{E}}_k(\Phi,A):=& \left\{\tilde{f} = \frac{f+M'+1}{\int f d\mu + M' + 1}: f\in \mathcal{E}_k(\Phi,A) \right\}\\
=&\left\{\tilde{f}(\cdot) = \frac{\theta_1 +M'+1 + \sum_{j=2}^\infty\theta_j\phi_j(\cdot)}{\theta_1 + M' + 1}: f(\cdot) = \sum_{j=1}^\infty\theta_j\phi_j(\cdot)\in \mathcal{E}_k(\Phi,A) \right\}\\
=&\left\{\tilde{f}(\cdot) = 1 + \sum_{j=2}^\infty\frac{\theta_j}{\theta_1 + M' + 1}\phi_j(\cdot): |\theta_j|<Aj^{-k} \right\}.
\end{align*}
This is a set of densities in $\mathcal{E}_k(\Phi,A')$, which is a subset $\tilde{\Theta}_k(\Phi,A',C')$ for some constant $A'$. Note that for $M'$ large, $A'\leq A$ and $C'\leq C$, which implies that $\tilde{\Theta}_k(\Phi,A',C')\subseteq \tilde{\Theta}_k(\Phi,A,C)$. Therefore, it suffices to establish a lower bound for $\tilde{\Theta}_k(\Phi,A',C')$.

Second, consider a subset of $\tilde{\mathcal{E}}_k(\Phi,A)$, denote by
\[
\mathcal{G}_1 := \left\{ g(\cdot) = 1 + \sum_{j=2}^\infty\theta_j\phi_j(\cdot):|\theta_j|<\tilde{A}j^{-k}\right\}.
\]
Note that $\mathcal{G}_1$ is a indeed a subset of $\tilde{\mathcal{E}}_k(\Phi,A)$ if, for example, $\tilde{A} = A/(A+M'+1)$. Moreover, we define
\[
\mathcal{G}_2 := \left\{ g(\cdot) = 1 + \sum_{j=1}^\infty\theta_j\phi_j(\cdot):|\theta_j|<A^*j^{-k}\right\}
\]
for some constant $A^*$. Note that we can change the index of $\theta_j$ by considering
\[
|\theta_j|<A^*j^{-k} \implies |\theta_j|<A^*\left(\frac{j}{j+1}\right)^{-k} (j+1)^{-k}\leq 2^kA^*(j+1)^{-k}.
\]
This implies that $\mathcal{G}_2$ is a subset of
\[
\mathcal{G}_3 := \left\{ g(\cdot) = 1 + \sum_{j=2}^\infty\theta_j\tilde{\phi}_j(\cdot):|\theta_j|<2^kA^*j^{-k}\right\}
\]
where $\tilde{\phi}_j = \phi_{j-1}$. Since $\{\phi_j\}$ is an orthonormal basis, it can be shown that the set $\mathcal{G}_1$ and $\mathcal{G}_3$ are isometric, which implies that they have the same order of entropy lower-bounded by the entropy of $\mathcal{G}_2$. Moreover, $\mathcal{G}_2$ and $\mathcal{E}_k(\Phi,A')$ are also isometric, so they also have the same order of $L^2([0,1],\mu)$ metric entropy.

Combining above results, we have found a subset of $\tilde{\Theta}_k(\Phi,A',C')$, $\mathcal{G}_1$, that has the same order of entropy as $\mathcal{E}_k(\Phi,A')$. Since $\tilde{\Theta}_k(\Phi,A',C')\subseteq \tilde{\Theta}_k(\Phi,A,C)$, by the results in Theorem \ref{thm:ta}, we conclude that
\[
M_2(\epsilon,\tilde{\Theta}_k(\Phi,A,C))\gtrsim \epsilon^{-2/(2k-1)}.
\]
On the other hand, since $\tilde{\Theta}_k(\Phi,A,C)\subseteq \mathcal{A}_k(\Phi,A,C)$, the upper bound follows. \qed

\subsection{Proof of Theorem \protect\ref{thm:minimax}}
First, we show the claim on the nonparametric regression. Given the entropy bounds from Theorem \ref{thm:entropy}, we can apply results from \cite{yangbarron97} (Theorem 9 and 10) and \cite{yangbarron99} (Theorem 6):
\[
\inf_{\hat{f}_n}\sup_{f\in\Theta_k(\Phi,A,C)}E\left[\|f-\hat{f}_n\|_2^2\right]\asymp n^{-\frac{2k-1}{2k}}.
\]
Second, we establish the minimax rates on the density subset $\tilde{\Theta}_k(\Phi,A,C)$. Again, for notational simplicity, we use $\|\cdot\|_2$ to denote the $L^2([0,1],\mu)$ norm, where $\mu$ is the Lebesgue measure on $[0,1]$. Note that we can not apply \cite{yangbarron99}'s results directly on our set $\tilde{\Theta}_k(\Phi,A,C)$ since it may not be bounded away from zero and it is not convex (due to reordering restrictions on the Fourier coefficients). Nevertheless, since $\mathcal{E}_k(\Phi,A) \subseteq \Theta_k(\Phi,A,C)\subseteq \mathcal{A}_k(\Phi,A,C)$, we have
\[
\tilde{\mathcal{E}}_k(\Phi,A) \subseteq \tilde{\Theta}_k(\Phi,A,C)\subseteq \tilde{\mathcal{A}}_k(\Phi,A,C)
\]
where (with some abuse of notation) $\tilde{\mathcal{E}}_k(\Phi,A)$, $\tilde{\Theta}_k(\Phi,A,C)$, and $\tilde{\mathcal{A}}_k(\Phi,A,C)$ are the density subsets of $\mathcal{E}_k(\Phi,A), \Theta_k(\Phi,A,C),\mathcal{A}_k(\Phi,A,C)$ respectively. Moreover, it can be verified from definition that $\tilde{\mathcal{E}}_k(\Phi,A)$ and $\tilde{\mathcal{A}}_k(\Phi,A,C)$ are both convex, so \cite{yangbarron99} applies to these two sets. In particular, the lower bound on the minimax rate of $\tilde{\Theta}_k(\Phi,A,C)$ is given by the a lower bound on $\tilde{\mathcal{E}}_k(\Phi,A)$ (the entropy of $\tilde{\mathcal{E}}_k(\Phi,A)$ is established in the proof of Corollary \ref{col:entropy}), and an upper bound is given by that of $\tilde{\mathcal{A}}_k(\Phi,A,C)$ (see Theorem 7 in \cite{yangbarron99}). Therefore, we have the following
\[
\inf_{\hat{f}_n}\sup_{f\in\tilde{\Theta}_k(\Phi,A,C)}E\left[\|f-\hat{f}_n\|_2^2\right]\asymp n^{-\frac{2k-1}{2k}}.
\]
\qed

\section{Proofs for Section 4}

\subsection{Proof of Theorem \protect\ref{thm:ta}}
Recall that $\lambda$ is defined as follows
\[
\lambda := \sqrt{\frac{\log(J)}{n}}\Phi^{-1}\left(1-\frac{1}{2\sqrt{2\pi}}\frac{1}{\sqrt{2\log(J)}}\frac{1}{J}\right)\max_{1\leq j\leq J}\left(\frac{1}{n}\sum_{i=1}^n \hat{Z}_{ij}^2\right)^\frac{1}{2}
\]
where $\hat{Z}_{ij}:= \phi_j(X_i) - n^{-1}\sum_{k=1}^n\phi_j(X_k)$. To establish the result, we need to bound $\lambda$ from below. Using the property of normal CDF that $1-\Phi(x)>1/(2\pi)^{1/2}(x/(x^2+1))\exp(-x^2/2)$ (see Lemma \ref{lm:gtail}), we have
\[
\Phi^{-1}\left(1-\frac{1}{2\sqrt{2\pi}}\frac{1}{\sqrt{2\log(J)}}\frac{1}{J}\right)> \sqrt{2\log(J)}.
\]
Now we bound $\max_{1\leq j\leq J}(n^{-1}\sum_{i=1}^n\hat{Z}_{ij}^2)^{1/2}$ from below. By triangle inequality, we have
\begin{align*}
\max_{1\leq j\leq J}\left(\frac{1}{n}\sum_{i=1}^n\hat{Z}_{ij}^2\right)^\frac{1}{2} &= \max_{1\leq j\leq J}\left(\frac{1}{n}\sum_{i=1}^n(\hat{Z}_{ij}-Z_{ij} + Z_{ij})^2\right)^\frac{1}{2}\\
& \geq \max_{1\leq j\leq J}\left(\frac{1}{n}\sum_{i=1}^nZ_{ij}^2\right)^\frac{1}{2} - \max_{1\leq j\leq J}\left(\frac{1}{n}\sum_{i=1}^n(\hat{Z}_{ij}-Z_{ij})^2\right)^\frac{1}{2}
\end{align*}
and we will bound each term separately.

First, following the step 1 in the proof of Theorem 2.4 in \cite{BCCHKN}, we have
\begin{align*}
P\left(\max_{1\leq j\leq J}\frac{1}{n}\sum_{i=1}^nZ_{ij}^2 \leq \frac{1}{2}\right) \leq P\left(\frac{1}{n}\sum_{i=1}^nZ_{ij}^2 \leq \frac{1}{2}\right)
\end{align*} 
where the inequality holds since $P(n^{-1}\sum_{i=1}^nZ_{ij}^2 \leq 1/2,\ \forall\ 1\leq j\leq J)\leq P(n^{-1}\sum_{i=1}^nZ_{ij}^2\leq 1/2)$ for any $1\leq j\leq J$. Then for a generic $1\leq j\leq J$ and some constant $c'>0$,
\begin{align*}
P\left(\frac{1}{n}\sum_{i=1}^nZ_{ij}^2\leq \frac{1}{2}\right) &\leq P\left(\frac{1}{n}\sum_{i=1}^nZ_{ij}^2\leq \frac{1}{2n}\sum_{i=1}^nE\left[Z_{ij}^2\right]\right)\\
& = P\left(\sum_{i=1}^nZ_{ij}^2-E\left[Z_{ij}^2\right]\leq -\frac{1}{2}\sum_{i=1}^nE\left[Z_{ij}^2\right]\right)\\
& \leq \exp\left(-\frac{(-\frac{1}{2}\sum_{i=1}^nE\left[Z_{ij}^2\right])^2}{2\sum_{i=1}^nE\left[Z_{ij}^4\right]}\right)\\
& \leq \exp\left(-\frac{n^2}{c'nM_J^2}\right)
\end{align*}
where the first inequality holds by the assumption that $E\left[Z_{ij}^2\right]\geq 1$, the second inequality holds by Bernstein inequality (see exercise 2.9 in \cite{Boucheron13}), and the last inequality holds by that $E\left[Z_{ij}^4\right] = E\left[(\phi_j(X)-E\left[\phi_j(X)\right])^4\right]\leq (2M_J)^2E\left[\phi_j^2(X)\right]\leq 4CM_J^2$. By assumption, $M_J^2\leq n/\log(n)$ and $J = n^p$, which implies $n/(c'M_J^2)\geq2\log(Jn)$. Then the above result implies that
\[
P\left(\frac{1}{n}\sum_{i=1}^nZ_{ij}^2\leq \frac{1}{2}\right) \leq \exp(-2\log(Jn)) = (Jn)^{-2}
\]
and in which case we have
\[
P\left(\max_{1\leq j\leq J}\frac{1}{n}\sum_{i=1}^nZ_{ij}^2 \leq \frac{1}{2}\right)\leq (Jn)^{-2}.
\]

Second, we bound the term $\max_{1\leq j\leq J}(n^{-1}\sum_{i=1}^n(\hat{Z}_{ij}-Z_{ij})^2)^{1/2}$. To simplify the notation, define
\[
\mathbb{Z}:= \max_{1\leq j\leq J}\left(\frac{1}{n}\sum_{i=1}^n(\hat{Z}_{ij}-Z_{ij})^2\right)^\frac{1}{2} = \max_{1\leq j\leq J}|\frac{1}{n}\sum_{i=1}^n\phi_j(X_i)-E\left[\phi_j(X_i)\right]|.
\]
By Bernstein inequality and maximal inequality (see Lemma \ref{lm:maximal} in the appendix),
\[
E\left[\mathbb{Z}\right] \leq \frac{K(M_J\log(J) + \sqrt{n\log(J)})}{n}.
\]
By Talagrand's inequality (see the version given by \cite{Bousquet2003Talagrand}),
\begin{align*}
P\left(\mathbb{Z}\geq E\left[\mathbb{Z}\right] +\sqrt{2\nu_nx} + \frac{Ux}{3}\right)\leq \exp(-x)
\end{align*}
where $U$ satisfies that $n^{-1}\|\phi_j(\cdot) - E\left[\phi_j(X_i)\right]\|_\infty\leq U <\infty$, and $\nu_n:= 2UE\left[\mathbb{Z}\right] + n\sigma^2$ where $\sigma^2:= n^{-2}\max_{1\leq j\leq J}E\left[(\phi_j(X_i)-E\left[\phi_j(X_i)\right])^2\right]$. Note that we can take $U:= 2M_J/n$ and $\sigma^2\leq C/n^2$. Then we have
\begin{align*}
&P\left(\mathbb{Z} \geq \frac{K(M_J\log(J)+ \sqrt{n\log(J)})}{n}\right.\\
&+\left. \sqrt{2\left(\frac{2M_J}{n}\frac{K(M_J\log(J) + \sqrt{n\log(J)})}{n}+ \frac{C}{n}\right)x} + \frac{2M_J}{n}x\right)\\
&\leq P\left(\mathbb{Z}\geq E\left[\mathbb{Z}\right] +\sqrt{2\nu_nx} + \frac{Ux}{3}\right)\leq \exp(-x).
\end{align*}
By Condition \ref{regularity}, $M_J^2 = o(n)$, which implies that there exists some $n_1\in\mathbf{N}$ such that for all $n\geq n_1$ and for $x=3\log(n)$ we have
\[
P\left(\mathbb{Z} \geq \frac{1}{2\sqrt{2}}\right) \leq \exp(-3\log(n)) = n^{-3}.
\]

Combining above, with probability at least $1-(Jn)^{-2} - n^{-3}$, we have the following 
\begin{align*}
\max_{1\leq j\leq J}\left(\frac{1}{n}\sum_{i=1}^n\hat{Z}_{ij}^2\right)^\frac{1}{2} & \geq \max_{1\leq j\leq J}\left(\frac{1}{n}\sum_{i=1}^nZ_{ij}^2\right)^\frac{1}{2} - \max_{1\leq j\leq J}\left(\frac{1}{n}\sum_{i=1}^n(\hat{Z}_{ij}-Z_{ij})^2\right)^\frac{1}{2}\\
& = \max_{1\leq j\leq J}\left(\frac{1}{n}\sum_{i=1}^nZ_{ij}^2\right)^\frac{1}{2}- \mathbb{Z}\\
&>\frac{1}{2\sqrt{2}}.
\end{align*}
To see this, note that
\begin{align*}
&P\left(\max_{1\leq j\leq J}\left(\frac{1}{n}\sum_{i=1}^n\hat{Z}_{ij}^2\right)^\frac{1}{2}> \frac{1}{2\sqrt{2}}\right)\\
\geq &P\left(\max_{1\leq j\leq J}\left(\frac{1}{n}\sum_{i=1}^nZ_{ij}^2\right)^\frac{1}{2}- \mathbb{Z}> \frac{1}{2\sqrt{2}}\right)\\
\geq &P\left(\max_{1\leq j\leq J}\left(\frac{1}{n}\sum_{i=1}^nZ_{ij}^2\right)^\frac{1}{2} > \frac{1}{\sqrt{2}}\ \text{and}\ \mathbb{Z}< \frac{1}{2\sqrt{2}}\right)\\
 = &1 - P\left(\max_{1\leq j\leq J}\left(\frac{1}{n}\sum_{i=1}^nZ_{ij}^2\right)^\frac{1}{2} \leq \frac{1}{\sqrt{2}}\ \text{or}\ \mathbb{Z}\geq \frac{1}{2\sqrt{2}}\right)\\
\geq &1- P\left(\max_{1\leq j\leq J}\left(\frac{1}{n}\sum_{i=1}^nZ_{ij}^2\right)^\frac{1}{2} \leq \frac{1}{\sqrt{2}}\right) - P\left(\mathbb{Z}\geq \frac{1}{2\sqrt{2}}\right)\\
\geq &1-(Jn)^{-2}-n^{-3}.
\end{align*}
This implies that
\begin{align*}
\lambda &= \sqrt{\frac{\log(J)}{n}}\Phi^{-1}\left(1-\frac{1}{2\sqrt{2\pi}}\frac{1}{\sqrt{2\log(J)}}\frac{1}{J}\right)\max_{1\leq j\leq J}(\frac{1}{n}\sum_{i=1}^n \hat{Z}_{ij}^2)^{1/2}\\
&\geq c'' \sqrt{\frac{\log^2(J)}{n}}
\end{align*}
with probability at least $1-(Jn)^{-2} - n^{-3}$ for $n\geq n_1$ and some constants $c''>0$.

Finally, to bound $\alpha_n$, we once again appeal to Talagrand's inequality. By assumption, $M_J^2 = o(n)$, then there exists $n_2\in\mathbf{N}$ such that for all $n\geq n_2$,
\begin{align*}
c'' \sqrt{\frac{\log^2(J)}{n}} \geq & \frac{K(M_J\log(J)+ \sqrt{n\log(J)})}{n}\\
+ & \sqrt{2\left(\frac{2M_J}{n}\frac{K(M_J\log(J) + \sqrt{n\log(J)})}{n}+ \frac{C}{n}\right)x} + \frac{2M_J}{n}x\\
\geq & E\left[\mathbb{Z}\right] + \sqrt{2\nu_nx} + \frac{Ux}{3}
\end{align*}
for $J = n^p$ and $x = 3\log(n)$. Then for $n\geq n^*:=\max\{n_1,n_2\}$,
\begin{align*}
P(\mathbb{Z}\geq\lambda)&\leq P\left(\mathbb{Z}\geq c'' \sqrt{\frac{\log^2(J)}{n}}\right) + (Jn)^{-2} + n^{-3}\\
&\leq P\left(\mathbb{Z}\geq E\left[\mathbb{Z}\right] + \sqrt{2\nu_n\log(n^3)} + \frac{U\log(n^3)}{3}\right) + (Jn)^{-2} + n^{-3}\\
&\leq (Jn)^{-2} + 2n^{-3}.
\end{align*}
Recall that $\mathbb{Z} = \max_{1\leq j\leq J}|n^{-1}\sum_{i=1}^n\phi_j(X_i)-E\left[\phi_j(X_i)\right]| = \max_{1\leq j\leq J}|\hat\theta_j-\theta_j|$, so the above result suggests that we can take $\alpha_n = (Jn)^{-2} + 2n^{-3}$. Note there are many other permissible choices of $\alpha_n$, and the proof can be modified accordingly if different $\alpha_n$'s are needed. \qed

\subsection{Proof of Theorem \protect\ref{thm:rates}}
To establish the result for the post processed $\hat{f}^*$ in \ref{p-alg}, we first establish the result for $\tilde{f}_J$ in (\ref{series:post}) and then use the results in \cite{Gajek86} to conclude. We keep $J$ explicit throughout the proof and we will substitute $J = n^p$ at the end of the proof. Moreover, in order to bound the size of selected set of indices $T$, for the convenience of the notation, we penalize using $2\lambda$. Note that this is without loss of generality as we can multiply the original $\lambda$ in Theorem \ref{thm:ta} by $1/2$, and the results in Theorem $\ref{thm:ta}$ still hold (just with different constants and $n^*$ in the proof).

The proof proceeds in six steps. In the first step, we decompose the MISE into the ``variance'' and ``bias'' components. We bound the cardinality of the set of selected indices in the second step. Next, we bound ``variance'' and ``bias'' separately by expressions of $\lambda$ and $\alpha_n$ in the third and fourth steps respectively. In step 5 we establish that the $\lambda$ and $\alpha_n$ we are using are the ``correct'' ones and we conclude in step 6.




\emph{Step 1: Decompose MISE.} Let $f_J(\cdot) := \sum_{j=1}^J \theta_j\phi_j(\cdot)$ be the infeasible estimator for $f$, where $\theta_j$'s are the true but unknown Fourier coefficients. Then we have
\begin{align*}
&E_f\left[\|f-\tilde{f}_J\|^2_{L^2}\right]\\
=&E_f\left[\|f-f_J\|^2_{L^2}\right] + E_f\left[\|\tilde{f}_J-f_J\|^2_{L^2}\right] + 2E_f\left[\|(f-f_J)(\tilde{f}_J-f_J)\|_{L^2}\right]\\
=& E_f\left[\|f-f_J\|^2_{L^2}\right] + E_f\left[\|\tilde{f}_J-f_J\|^2_{L^2}\right]\\
 = & \sum_{j = J+1}^\infty \theta_j^2 + E\left[\sum_{j\in T}(\hat\theta_j - \theta_j)^2 + \sum_{j\in T^c}\theta_j^2 \right]\\
 = & E\left[\sum_{j\in T}(\hat\theta_j - \theta_j)^2\right] +E\left[\sum_{j\in T^c}\theta_j^2 + \sum_{j=J+1}^\infty\theta_j^2 \right]
\end{align*}
where the second and third equalities follow from orthonormality. Note that the last line corresponds to the variance-bias trade-off; however, the randomness of the set $T$ of the selected indices no longer allows the interchange of expectation and summation.

\emph{Step 2: Cardinality of Selected Indices.} Recall the definition of the selected indices $T$:
\[
T = \{j\in\{1,\cdots,J\}: \hat\theta_j \geq 2\lambda \}
\]
where $\lambda \geq (1-\alpha_n)$-quantile of $\|\hat\theta^J-\theta^J\|_\infty$. Then for $j\in T$, with probability at least $1-\alpha_n$,
\[
2\lambda \leq |\hat\theta_j| \leq |\hat\theta_j - \theta_j| + |\theta_j| \leq \lambda + Aj^{-k}
\]
where the last inequality holds by the definition of approximate sparsity set. This implies that for $j\in T$, with probability at least $1-\alpha_n$, 
\[
\lambda \leq Aj^{-k} \iff j \leq A^\frac{1}{k}\lambda^{-\frac{1}{k}}.
\]
That is, $T \subseteq \{1\leq j\leq J: j\leq A^{1/k}\lambda^{-1/k}\}$ with probability at least $1-\alpha_n$. This result establishes that with probability at least $1-\alpha_n$, the cardinality of the set of selected indices $T$ satisfies
\[
|T| \leq A^\frac{1}{k}\lambda^{-\frac{1}{k}}.
\]

\emph{Step 3: Bound on Bias.} To control the bias term $E\left[\sum_{j\in T^c}\theta_j^2 + \sum_{j=J+1}^\infty\theta_j^2 \right]$, we need to control the random set $T^c$, the set of non-selected indices. By triangle inequality, we have
\[
|\theta_j|\leq |\hat\theta_j| + |\hat\theta_j - \theta_j|
\]
Note that since $\lambda \geq (1-\alpha_n)$-quantile of $\|\hat\theta^J-\theta^J\|_\infty$, we have $|\hat\theta_j - \theta_j|\leq\lambda$ with probability at least $1-\alpha_n$. Moreover,
by definition, for $j\in T^c$, $|\hat\theta_j|< 2\lambda$. Combining above, on $T^c$,
\[
|\theta_j|\leq 3\lambda
\]
with probability at least $1-\alpha_n$. This result will help us control the bias on the random set $T^c$. Let $E_n$ denote the event that
\[
E_n := \{\lambda\geq \max_{j\leq J}|\hat\theta_j - \theta_j|\}
\]
and note that the probability of this event $E_n$ happening is at least $1-\alpha_n$, and on this event, $|\theta_j|\leq 3\lambda$ for $j\in T^c$. We will show that $\alpha_n$ can go to zero sufficiently fast for us to establish minimax rate. Then we can establish the following: for some constants $C_1,C_2,C_3$,
\begin{align*}
&E\left[\sum_{j\in T^c}\theta_j^2 + \sum_{j=J+1}^\infty\theta_j^2 \right]\\
&\leq E\left[\sum_{j\in T^c\cap j\leq m}\theta_j^2 + \sum_{j=m+1}^\infty\theta_j^2 \right]\\
& = E\left[\left(\sum_{j\in T^c\cap j\leq m}\theta_j^2 + \sum_{j=m+1}^\infty\theta_j^2\right)\cdot \mathbf{1}\{E_n\} \right] + E\left[\left(\sum_{j\in T^c\cap j\leq m}\theta_j^2 + \sum_{j=m+1}^\infty\theta_j^2\right)\cdot \mathbf{1}\{E_n^c\}\right ]\\
&\leq E\left[9m\lambda^2 + C_1m^{-2k+1}\right] + C_2\cdot P(E_n^c)
\end{align*}
where the first inequality holds for every $m\leq J$; the second inequality holds by that on the event $E_n$, $|\theta_j|\leq 3\lambda$, by H\"older's inequality, and by that the true density is bounded. Since the above expression holds for all $m\leq J$, we have the following
\begin{align*}
&E\left[\sum_{j\in T^c}\theta_j^2 + \sum_{j=J+1}^\infty\theta_j^2 \right]\\
&\leq E\left[\min_m 9m\lambda^2 + C_1m^{-2k+1}\right] + C_2\cdot P(E_n^c)\\
&\leq C_3E\left[\lambda^{\frac{2k-1}{k}}\right] + \alpha_n C_2
\end{align*}
where the last inequality holds by solving the minimization problem over $m$ and by that the probability of event $E_n^c$ is at most $\alpha_n$. We will specify $\alpha_n$ and bound $E\left[\lambda^{\frac{2k-1}{k}}\right]$ explicitly.

\emph{Step 4: Bound on Variance.} In this section, we establish bounds on $E\left[\sum_{j\in T}(\hat\theta_j-\theta_j)^2\right]$. Recall $E_n$ is the event that $E_n= \{\lambda\geq\max_{j\in J}|\hat\theta_j - \theta_j|\}$, and in step 2 we have shown that $|T|\leq A^{1/k}\lambda^{-1/k}$ on this event. Then
\begin{align*}
 &E\left[\sum_{j\in T}(\hat\theta_j-\theta_j)^2\right]\\
=&E\left[\sum_{j\in T}(\hat\theta_j-\theta_j)^2\mathbf{1}\{E_n\}\right] + E\left[\sum_{j\in T}(\hat\theta_j-\theta_j)^2\mathbf{1}\{E_n^c\}\right]\\
\leq & E\left[|T| \max_{j\in T} (\hat\theta_j-\theta_j)^2\mathbf{1}\{E_n\}\right] + \sum_{j=1}^JE\left[(\hat\theta_j - \theta_j)^2\mathbf{1}\{E_n^c\}\right]
\end{align*}
where the inequality holds since (i) $\sum_{j\in T}(\hat\theta_j-\theta_j)^2\mathbf{1}\{E_n\} \leq |T| \max_{j\in T} (\hat\theta_j-\theta_j)^2\mathbf{1}\{E_n\}$ over the sample space; (ii) we can bound the sum over random $T\subseteq \{1,\cdots, J\}$ from above with the deterministic sum over $1\leq j\leq J$, and then interchange the expectation and summation.

First, we bound $E\left[|T| \max_{j\in T} (\hat\theta_j-\theta_j)^2\mathbf{1}\{E_n\}\right]$ with the help of set $E_n$:
\begin{align*}
&E\left[|T| \max_{j\in T} (\hat\theta_j-\theta_j)^2\mathbf{1}\{E_n\}\right]\\
\leq & A^\frac{1}{k}E\left[\lambda^{-\frac{1}{k}} \max_{j\in T} (\hat\theta_j-\theta_j)^2\mathbf{1}\{E_n\}\right]\\
\leq & A^\frac{1}{k}E\left[\lambda^{-\frac{1}{k}}\lambda^2\right]\\
= & A^\frac{1}{k}E\left[\lambda^{\frac{2k-1}{k}}\right]
\end{align*}
where the first inequality holds since $|T| \leq A^{1/k}\lambda^{-1/k}$ on $E_n$ and the second inequality holds since on $E_n$, $\max_{j\in J}|\hat\theta_j - \theta_j|\leq\lambda$.

Second, we bound $\sum_{j=1}^JE\left[(\hat\theta_j - \theta_j)^2\mathbf{1}\{E_n^c\}\right]$. By Cauchy-Schwarz,
\[
\sum_{j=1}^JE\left[(\hat\theta_j - \theta_j)^2\mathbf{1}\{E_n^c\}\right]
\leq \sum_{j=1}^J\left(E\left[(\hat\theta_j - \theta_j)^4\right]\right)^\frac{1}{2}(P(E_n^c))^\frac{1}{2}.
\]
Moreover, for each $j$,
\begin{align*}
E\left[(\hat\theta_j - \theta_j)^4\right] &= E\left[\left(\frac{1}{n}\sum_{i=1}^n\phi_j(X_i) -E\left[\phi_j(X_i)\right]\right)^4\right]\\
& = n^{-4}E\left[\left(\sum_{i=1}^n\phi_j(X_i) -E\left[\phi_j(X_i)\right]\right)^4\right]\\
& \leq c_1 n^{-4} E\left[\left(\sum_{i=1}^n\left(\phi_j(X_i) -E\left[\phi_j(X_i)\right]\right)^2\right)^2\right]\\
& = c_1 n^{-2} E\left[\left(\frac{1}{n}\sum_{i=1}^n\left(\phi_j(X_i) -E\left[\phi_j(X_i)\right]\right)^2\right)^2\right]\\
& \leq c_1 n^{-2} E\left[\frac{1}{n}\sum_{i=1}^n \left(\phi_j(X_i)-E\left[\phi_j(X_i)\right]\right)^4 \right]\\
& = c_1n^{-2} E\left[(\phi_j(X_i)-E\left[\phi_j(X_i)\right])^4 \right]
\end{align*}
where the first inequality holds by Marcinkiewicz-Zygmund inequality for some constant $c_1>0$, the second inequality holds by convexity of the function $x\mapsto x^2$ and Jensen's inequality, and the last equality holds by i.i.d assumption. This derivation suggests that we need to bound the fourth central moment of $\phi_j(X_i)$. If the basis is uniformly bounded, the result is trivial. On the other hand, if the basis grows, i.e. $\max_{1\leq j\leq J}\|\phi_j(\cdot)\|_\infty \leq M_J$, we can bound $E\left[\phi_j^4(X_i)\right]$ explicitly. Note that for some constant $C>0$,
\begin{align*}
E\left[\phi_j^4(X_i)\right] \leq M_J^2E\left[\phi_j^2(X_i)\right] = M_J^2 \int \phi_j^2(x)f(x)dx \leq CM_J^2
\end{align*}
where the last inequality hold by orthonormality and by that $f$ is bounded. This gives us
\[
E\left[(\hat\theta_j-\theta_j)^4\right]\leq C' M_J^2/n^2
\]
for some constant $C'>0$. Therefore, for some constant $c_2>0$,
\begin{align*}
\sum_{j=1}^JE\left[(\hat\theta_j - \theta_j)^2\mathbf{1}\{E_n^c\}\right]
\leq \sum_{j=1}^J\left(E\left[(\hat\theta_j - \theta_j)^4\right]\right)^\frac{1}{2}(P(E_n^c))^\frac{1}{2} \leq c_2 JM_J/n\alpha_n^{1/2}
\end{align*}
where the second inequality holds by that $P(E_n^c)\leq\alpha_n$. We need $JM_J/n\alpha_n^{1/2}$ to go to zero sufficiently fast (at least as fast as the minimax rate), which is assumed and can be verified for a given orthonormal basis. Combining above, we have an upper bound on variance
\[
E\left[\sum_{j\in T}(\hat\theta_j-\theta_j)^2\right] \leq c_1E\left[\lambda^{\frac{2k-1}{k}}\right] + c_2JM_J/n\alpha_n^{1/2}.
\]

\emph{Step 5: The ``Right'' Penalization Parameter.} We show the $\lambda$ proposed in (\ref{ta:lambda}) will give us the ``correct'' minimax rate. Recall that $\lambda$ is defined as follows
\[
\lambda = \sqrt{\frac{\log(J)}{n}}\Phi^{-1}\left(1-\frac{1}{2\sqrt{2\pi}}\frac{1}{\sqrt{2\log(J)}}\frac{1}{J}\right)\max_{1\leq j\leq J}\left(\frac{1}{n}\sum_{i=1}^n \hat{Z}_{ij}^2\right)^\frac{1}{2}
\]
where $\hat{Z}_{ij}:= \phi_j(X_i) - n^{-1}\sum_{k=1}^n\phi_j(X_k)$. Then
\begin{align*}
&E\left[\lambda^\frac{2k-1}{k}\right]\\
&= E\left[\left(\sqrt{\frac{\log(n)}{n}}\Phi^{-1}\left(1-\frac{1}{2\sqrt{2\pi}}\frac{1}{\sqrt{2\log(J)}}\frac{1}{J}\right)\max_{1\leq j\leq J}(\frac{1}{n}\sum_{i=1}^n \hat{Z}_{ij}^2)^\frac{1}{2}\right)^\frac{2k-1}{k}\right]\\
&\leq \left(\frac{\log(n)\log(2\sqrt{2\pi}\sqrt{2\log(J)}J)}{n}\right)^\frac{2k-1}{2k}E\left[\left(\max_{1\leq j\leq J}\left(\frac{1}{n}\sum_{i=1}^n \hat{Z}_{ij}^2\right)^\frac{1}{2}\right)^\frac{2k-1}{k}\right]\\
&\leq \left(\frac{\log(n)\log(2\sqrt{2\pi}\sqrt{2\log(J)}J)}{n}\right)^\frac{2k-1}{2k}\left(E\left[\max_{1\leq j\leq J}\frac{1}{n}\sum_{i=1}^n \hat{Z}_{ij}^2\right]\right)^\frac{2k-1}{2k}\\
&\leq C'\left(\frac{\log(n)\log(J)}{n}\right)^{\frac{2k-1}{2k}}\left(E\left[\max_{1\leq j\leq J}\left|\frac{1}{n}\sum_{i=1}^n \phi_j^2(X_i)-E\left[\phi_j^2(X_i)\right]\right|\right] + C''\right)^\frac{2k-1}{k}\\
&\leq C'\left(\frac{\log(n)\log(J)}{n}\right)^{\frac{2k-1}{2k}}\left(\frac{M_J^2\log(J) + \sqrt{nM_J^2\log(J)}}{n} + C''\right)^\frac{2k-1}{k}\\
&\leq C\left(\frac{\log(n)\log(J)}{n}\right)^{\frac{2k-1}{2k}}
\end{align*}
where the first inequality holds by the property that $\Phi^{-1}(1-x)\leq\sqrt{2\log(1/x)}$ (see Lemma \ref{lm:gtail}), the second inequality holds by Jensen's inequality, and the third inequality holds by the definition of $\hat{Z}_{ij}$ and by the fact that the density is bounded so that by orthonormality $E\left[\phi_j^2(X_i)\right]\leq C''$ for some constant $C''>0$, the fourth inequality holds by maximal inequality (see Lemma \ref{lm:maximal} in the appendix), and the last inequality holds by the assumption that $M_J^2\leq n/\log(n)$ and $J=n^p$ for some constant $C>0$.

\emph{Step 6: Conclusion.} By Theorem \ref{thm:ta}, under Condition \ref{regularity}, there exists $n^*\in\mathbf{N}$ such that for all $n\geq n^*$,
\[
\lambda\geq (1-\alpha_n)-\text{quantile of}\ \max_{1\leq j\leq J}|\hat\theta_j - \theta_j|
\]
with $\alpha_n = (Jn)^{-2} + 2n^{-3}$. Then combining results from step 1 to step 5, by the assumptions on $J$ and $M_J$, we have for all $n\geq n^* = \max\{n_1,n_2\}$,
\begin{align*}
&E_f\left[\|f-\tilde{f}_J\|_{L_2}^2\right]\\ 
&\leq E\left[\sum_{j\in T}(\hat\theta_j - \theta_j)^2\right] +E\left[\sum_{j\in T^c}\theta_j^2 + \sum_{j=J+1}^\infty\theta_j^2 \right]\\
&\leq c_1E\left[\lambda^{\frac{2k-1}{k}}\right] + c_2\frac{JM_J}{n}\alpha_n^{\frac{1}{2}} + C_3E\left[\lambda^{\frac{2k-1}{k}}\right] + \alpha_nC_2\\
&\leq \tilde{C}\left(\frac{\log(n)\log(J)}{n}\right)^{\frac{2k-1}{2k}} + c_2\frac{JM_J}{n}\left((Jn)^{-2} + 2n^{-3}\right)^{\frac{1}{2}} + C_2\left((Jn)^{-2} + 2n^{-3}\right)\\
&= O\left(\frac{\log^2(n)}{n}\right)^{\frac{2k-1}{2k}}
\end{align*}
where the last inequality holds by that $\alpha_n = (Jn)^{-2} + 2n^{-3}< n^{-(2k-1)/(2k)}$ and by the assumption that $(JM_J/n)\alpha_n^{1/2}\leq n^{-(2k-1)/(2k)}$. By \cite{Gajek86}, the post-processed $\hat{f}^*$ has smaller MISE, and note that our proof does not depend on a specific $f\in \tilde\Theta_k(\Phi,A,C)$, which proves the desired result
\[
\sup_{f\in\tilde\Theta_k(\Phi,A,C)} E_f\left[\|f-\hat{f}^*\|_{L_2}^2\right]\leq \sup_{f\in\tilde\Theta_k(\Phi,A,C)} E_f\left[\|f-\tilde{f}_J\|_{L_2}^2\right] = O\left(\left(\frac{\log^2(n)}{n}\right)^{\frac{2k-1}{2k}}\right).
\]
\qed

\section{Additional Technical Results}
\begin{lemma}\label{lm:inclusion}
Let $A>0$ and $k>1/2$ be some constants. Let the constant $C$ in the definition of $\Theta_k(\Phi,A,C)$ and $\mathcal{A}_k(\Phi,A,C)$ be such that $C\geq A^2/(2k-1)$. Then
\begin{equation}
\mathcal{E}_k(\Phi,A)\subseteq\Theta_k(\Phi,A,C)\subseteq\mathcal{A}_k(\Phi,A,C)
\end{equation}
\end{lemma}

\begin{proof}[Proof of Lemma \ref{lm:inclusion}]
First, note that by definition, $\mathcal{E}_k(\Phi,A)$ is a special case of $\Theta_k(\Phi,A,C)$ without reordering. In particular, in the definition of $\mathcal{E}_k(\Phi,A)$, since $|\theta_j|<Aj^{-k}$, then
\[
\sum_{j= J+1}^\infty\theta_j^2 \leq A^2\int_J^\infty t^{-2k}dt = \frac{A^2}{2k-1}J^{-2k+1}
\]
This establishes $\mathcal{E}_k(\Phi,A)\subseteq \Theta_k(\Phi,A,C)$.

Moreover, the restrictions on the tail sum $\sum_{j= J+1}^\infty\theta_j^2$ are identical in $\Theta_k(\Phi,A,C)$ and $\mathcal{A}_k(\Phi,A,C)$. Additional restrictions on individual $\theta_j$ makes $\Theta_k(\Phi,A,C)$ a subset of $\mathcal{A}_k(\Phi,A,C)$. \end{proof}

\begin{lemma}\label{lm:gtail}
Let $\Phi$ denote the CDF of the standard normal random variable, then for $x\geq 0$,
\begin{align*}
\frac{1}{\sqrt{2\pi}}\frac{x}{x^2+1}\exp\left(-\frac{x^2}{2}\right)< 1- \Phi(x) < \frac{1}{\sqrt{2\pi}}\frac{1}{x}\exp\left(-\frac{x^2}{2}\right).
\end{align*}
\end{lemma}

\begin{proof}[Proof of Lemma \ref{lm:gtail}]
To show the upper bound, note that for $x\geq 0$
\begin{align*}
1-\Phi(x) &= \int_x^\infty\phi(s)ds\\
& = \frac{1}{\sqrt{2\pi}}\int_x^\infty \exp\left(-\frac{s^2}{2}\right) ds\\
& < \frac{1}{\sqrt{2\pi}}\int_x^\infty\frac{s}{x}\exp\left(-\frac{s^2}{2}\right) ds\\
& = \frac{1}{\sqrt{2\pi}}\frac{1}{x}\exp\left(-\frac{x^2}{2}\right).
\end{align*}
To show the lower bound, let
\[
h(x) := 1-\Phi(x) - \frac{1}{\sqrt{2\pi}}\frac{x}{x^2+1}\exp\left(-\frac{x^2}{2}\right).
\]
Since $h(0) >0$, $h'(x)<0$ for all $x\geq 0$, and $h(x)\to 0$ as $x\to\infty$, we must have $h(x)>0$ for all $x\geq 0$. This gives us the lower bound.
\end{proof}

\begin{lemma}\label{lm:maximal}
Let $\{X_i\}_{i=1}^n\sim X$ be an i.i.d sample with $X\in [0,1]\subseteq \mathbf{R}$. Suppose the density $f$ of $X$ is bounded above by some constant $\tilde{C}$ and that the orthonormal basis $\{\phi_j\}_{j=1}^\infty$ of $L^2([0,1],\mu)$ is such that $\max_{1\leq j\leq J}\|\phi_j(\cdot)\|_\infty \leq M_J$ for some $M_J$ that potentially grows with $J$ for all $J$. Then for some constants $K_1$ and $K_2$,
\begin{align*}
E\left[\max_{j\leq J} \left|\sum_{i=1}^n\phi_j(X_i) - E\left[\phi_j(X_i)\right] \right|\right]\leq K_1\left(M_J\log(J) + \sqrt{n}\sqrt{\log(J)} \right)
\end{align*}
and
\begin{align*}
E\left[\max_{j\leq J} \left|\sum_{i=1}^n\phi_j^2(X_i) - E\left[\phi_j^2(X_i)\right] \right|\right]\leq K_2\left(M_J^2\log(J) + \sqrt{nM_J^2}\sqrt{\log(J)} \right).
\end{align*}

\end{lemma}

\begin{proof}[Proof of Lemma \ref{lm:maximal}]
First, since $\max_{1\leq j\leq J}\|\phi_j(\cdot)\|_\infty \leq M_J$, we have 
\[
|\phi_j(X_i) - E\left[\phi_j(X_i)\right]|\leq 2M_J.
\]
Since the density is bounded by $\tilde{C}$,
\[
\max_{j\leq J} E\left[\phi_j^2(X)\right] = \max_{j\leq J} \int \phi_j^2(x)f(x) dx \leq \max_{j\leq J} \tilde{C}\int \phi_j^2(x) dx = \tilde{C}
\]
where the last equality holds by orthonormality. This implies
\[
Var\left(\sum_{i=1}^n\phi_j(X_i) - E\left[\phi_j(X_i)\right]\right) = \sum_{i=1}^n Var(\phi_j(X_i))\leq nE\left[\phi_j^2(X_i)\right] \leq n\tilde{C}.
\]
Then, by Bernstein's inequality (Lemma 2.2.9 in \cite{VW96}), 
\[
P\left(\left|\sum_{i=1}^n\phi_j(X_i) - E\left[\phi_j(X_i)\right] \right|>x\right)\leq 2\exp\left(-\frac{1}{2}\frac{x^2}{n\tilde{C} +\frac{2M_Jx}{3}}\right)
\]
and by maximal inequality (Lemma 2.2.10 in \cite{VW96}),
\begin{align*}
E\left[\max_{j\leq J} \left|\sum_{i=1}^n\phi_j(X_i) - E\left[\phi_j(X_i)\right] \right|\right]\leq K_1\left(M_J\log(J) + \sqrt{n}\sqrt{\log(J)} \right)
\end{align*}
for some fixed constant $K_1$.

The second part of the statement can be shown using similar arguments. By assumption, we have $\max_{1\leq j\leq J}\|\phi_j^2\|_\infty \leq M_J^2$ for some potentially growing $M_J$ for all $J$. Then
\begin{align*}
&Var\left(\sum_{i=1}^n\phi_j^2(X_i) -E\left[\phi_j(X_i)^2\right]\right) \\
&= nVar(\phi_j^2(X))\leq nE\left[\phi_j^4(X)\right]\leq nM_J^2E\left[\phi_j(X)^2\right]\leq n\tilde{C}M_J^2.
\end{align*}
Then by Bernstein's inequality 
\[
P\left(\left|\sum_{i=1}^n\phi_j^2(X_i) - E\left[\phi_j^2(X_i)\right]\right|>x \right)\leq 2\exp\left(-\frac{1}{2}\frac{x^2}{n\tilde{C}M_J^2+ \frac{2M_J^2x}{3}}\right)
\]
which holds for all $1\leq j\leq J$, and by Maximal inequality,
\[
E\left[\max_{j\leq J}\left|\sum_{i=1}^n\phi_j^2(X_i) - E\left[\phi_j^2(X_i)\right]\right|\right]\leq K_2\left(M_J^2\log(J) + \sqrt{nM_J^2}\sqrt{\log(J)}\right)
\]
for some fixed constant $K_2$.
\end{proof}

\end{appendix}

\section*{Acknowledgements}
The author would like to express his gratitude to Andres Santos and Denis Chetverikov for their generous time and extremely helpful discussions that have lead to substantial improvements to this project. The author would also like to thank Rosa Matzkin, Oscar H. Madrid-Padilla, Mingli Chen, Rodrigo Pinto, and participants at UCLA econometrics proseminars for their helpful suggestions.

\end{document}